\documentclass[A4,10pt]{iopart}
\usepackage{amsfonts}
\usepackage{amssymb}
\usepackage{geometry}
\usepackage{graphicx}

\newtheorem{theorem}{Theorem}
\newtheorem{proposition}{Proposition}
\newtheorem{lemma}{Lemma}

\newtheorem{definition}{Definition}
\newtheorem{corollary}{Corollary}
\newtheorem{remark}{Remark}
\newenvironment{proof}[1][Proof]{\noindent\textbf{#1.} }{\ \rule{0.5em}{0.5em}}

\newtheorem{example}{Example}

\begin{document}

\title{Curvature of fluctuation geometry and its implications on Riemannian fluctuation theory}
\author{L. Velazquez}
\address{Departamento de F\'\i sica, Universidad Cat\'olica del Norte, Av. Angamos 0610, Antofagasta,
Chile.}

\begin{abstract}
Fluctuation geometry was recently proposed as a counterpart approach of Riemannian geometry of inference theory (widely known as \emph{information geometry}). This theory describes the geometric features of the statistical manifold $\mathcal{M}$ of random events that are described by a family of continuous distributions $dp(x|\theta)$. A main goal of this work is to clarify the statistical relevance of Levi-Civita curvature tensor $R_{ijkl}(x|\theta)$ of the statistical manifold $\mathcal{M}$. For this purpose, the notion of \emph{irreducible statistical correlations} is introduced. Specifically, a distribution $dp(x|\theta)$ exhibits irreducible statistical correlations if every distribution $dp(\check{x}|\theta)$ obtained from $dp(x|\theta)$ by considering a coordinate change $\check{x}=\phi(x)$ cannot be factorized into independent distributions as $dp(\check{x}|\theta)=\prod_{i}dp^{(i)}(\check{x}^{i}|\theta)$. It is shown that the curvature tensor $R_{ijkl}(x|\theta)$ arises as a direct indicator about the existence of irreducible statistical correlations. Moreover, the curvature scalar $R(x|\theta)$ allows to introduce a criterium for the applicability of the \emph{gaussian approximation} of a given distribution function. This type of asymptotic result is obtained in the framework of the second-order geometric expansion of the distributions family $dp(x|\theta)$, which appears as a counterpart development of the high-order asymptotic theory of statistical estimation. 

In physics, fluctuation geometry represents the mathematical apparatus of a Riemannian extension for Einstein's fluctuation theory of statistical mechanics. Some exact results of fluctuation geometry are now employed to derive the \emph{invariant fluctuation theorems}. Moreover, the curvature scalar allows to express some asymptotic formulae that account for the system fluctuating behavior beyond the gaussian approximation, e.g.: it appears as a second-order correction of Legendre transformation between thermodynamic potentials, $P(\theta)=\theta_{i}\bar{x}^{i}-s(\bar{x}|\theta)+k^{2}R(x|\theta)/6$.
\newline
\newline
PACS numbers: 02.50.-r; 02.40.Ky; 05.45.-a; 02.50.Tt \newline
Keywords: Geometrical methods in statistics; Fluctuation theory
\end{abstract}

\section{Introduction}

While attempting to obtain more general fluctuation theorems for systems in thermodynamic equilibrium, Curilef and I have discovered a remarkable analogy between fluctuation theory and inference theory \cite{Vel.GEFT}. The analysis of this analogy and its related mathematical questions motivates, by itself, a revision of foundations of classical statistical mechanics. For example, these statistical theories exhibit certain inequalities in the form of \emph{uncertainty-like relations}, e.g.: particular forms of Cramer-Rao theorem and their counterparts in the framework of fluctuation theory \cite{Vel.Book}. Such inequalities provide a strong support to Bohr's conjecture about the existence of complementary quantities in any physical theory with a statistical apparatus \cite{Bohr}-\cite{Uffink}. This viewpoint was employed in Ref.\cite{Vel.CSM} to proposed a reformulation of principles of classical statistical mechanics starting from the notion of complementarity.

Recently, the same analogy was considered in Ref.\cite{Vel.GEO} to propose a Riemannian extension of Einstein's fluctuation theory. The mathematical apparatus of this development is the \emph{Riemannian geometry of fluctuation theory}, which is hereinafter referred to as \emph{fluctuation geometry} \cite{Vel.FG}. Roughly speaking, fluctuation geometry constitutes a counterpart approach of Riemannian geometry of inference theory in the framework of continuous distributions \cite{Amari}\footnote{Riemannian geometry of inference theory is widely known in the literature as \emph{information geometry} or \emph{Riemannian geometry on statistical manifolds} \cite{Amari}. Nevertheless, the denomination \emph{inference geometry} was previously adopted in Ref.\cite{Vel.FG} to avoid the ambiguity with \emph{fluctuation geometry} and emphasize the existing connections between these two developments. In my opinion, denominations as \emph{information geometry} or \emph{Riemannian geometry on statistical manifolds} can equivalently apply for both Riemannian geometries of fluctuation theory and inference theory.}. I understand that this form of statistical geometry is previously unknown in the literature, so that, its study was recently addressed from an axiomatic perspective in Ref.\cite{Vel.FG}. This paper represents a continuation of this previous work.

Present contribution is devoted to deepen on mathematical aspects and physical implications of fluctuation geometry, in particular, to clarify the statistical relevance of the curvature of this Riemannian geometry and discuss new implications on Riemannian extension of Einstein's fluctuation theory. Previously \cite{Vel.FG}, I conjectured that the curvature notion of fluctuation geometry should account for the existence of \emph{irreducible statistical correlations}. The validity of this conjecture will be analyzed in this work, as well as the role of curvature tensor in the second-order geometric expansion of a continuous distribution. For the sake of self-consistence, this study is preceded by an introduction to fluctuation geometry, which is devoted to discuss some key concepts and results of this statistical development.

\section{An introduction to fluctuation geometry}\label{Antecedents}

\subsection{Statistical manifolds $\mathcal{M}$ and $\mathcal{P}$ and their coordinate representations}

Let us denote by $\mathcal{M}$ a certain universe of random events, and by $\epsilon$ an elementary event of $\mathcal{M}$. The event $\epsilon$ is \emph{elementary} because of the occurrence of any event $\mathcal{A}\in\mathcal{M}$ implies either the occurrence of the event $\epsilon$ or its non-occurrence. As expected, any general event $\mathcal{A}\in\mathcal{M}$ can be regarded as a subset of elementary events. Hereinafter, only elementary events are considered, so that, any elementary event $\epsilon$ will be simply referred to as an event. Let us consider that behavior of random events depends on certain external conditions, which are denoted by $\mathfrak{E}$. The universe of all admissible external conditions $\mathfrak{E}$ constitute a second abstract space $\mathcal{P}$, the space of external conditions. Hereinafter, let us admit that the universe of random events $\mathcal{M}$ and the space of external conditions $\mathcal{P}$ also represent smooth manifolds that are endowed of a \emph{differential structure}. In other words, $\mathcal{M}$ and $\mathcal{P}$ are \emph{differentiable manifolds} (they are locally similar enough to real spaces to allow the development of differential and integral calculus). For the sake of convenience, let us assume that the manifold $\mathcal{M}$ ($\mathcal{P}$) exhibits a diffeomorphism with the real space $\mathbb{R}^{n}$ ($\mathbb{R}^{m}$).

Let us consider that the manifolds of random events $\mathcal{M}$ and the external conditions $\mathcal{P}$ are \emph{abstract mathematical objects}. From the physical viewpoint, one can perform a quantitative characterization about the occurrence of a given event $\epsilon$ throughout measuring of certain \emph{observable quantities}. Of course, any observable that is measured in this context is a \emph{random quantity}. Mathematically speaking, a random quantity is defined as a \emph{real function} $\sigma(\epsilon)$ \emph{of random events}, that is, a map of the statistical manifold $\mathcal{M}$ on the one-dimensional real space $\mathbb{R}$, $\sigma:\mathcal{M}\rightarrow \mathbb{R}$. Let us now consider a set of $n$ independent random quantities $\xi=(\sigma^{1}, \sigma^{2}, \ldots \sigma^{n} )$, and denote by $x=(x^{1},x^{2},\ldots x^{n})$ a certain set of their admissible values. It is said that the set of random quantities $\xi$ is \emph{complete} when the same one constitutes a diffeomorphism $\xi: \mathcal{M}\rightarrow \mathcal{R}_{x}$ between the statistical manifold $\mathcal{M}$ and a certain subset $\mathcal{R}_{x}\subset \mathbb{R}^{n}$\footnote{Notice that a \emph{complete set of random quantities} $\xi$ exhibits a certain analogy with the notion of \emph{complete set of commuting observables} in quantum mechanics, which allows a univocal definition for the state $\Psi$ of a certain system.}. The real subset $\mathcal{R}_{x}$ will be regarded as a \emph{coordinate representation} of the manifold $\mathcal{M}$ in the $n$-dimensional real space $\mathbb{R}^{n}$, while $x=(x^{1},x^{2},\ldots x^{n})$ denotes the coordinates of a certain event $\epsilon\in\mathcal{M}$. Analogously, let us also assume that any realization of the external conditions $\mathfrak{E}$ can be parameterized by a set of continuous real variables $\theta=(\theta^{1},\theta^{2},\ldots,\theta^{m})$ that belong to a subset $\mathcal{R}_{\theta}$ of the $m$-dimensional real space $\mathbb{R}^{m}$. Let us suppose that the correspondence $\theta:\mathcal{P}\rightarrow\mathcal{R}_{\theta}$ represents a diffeomorphism. Hereinafter, the real subset $\mathcal{R}_{\theta}$ will be regarded as a coordinate representation of the manifold $\mathcal{P}$.

One can perform an indirect but complete characterization of the abstract statistical manifolds $\mathcal{M}$ and $\mathcal{P}$ studying the behavior of a complete set of random quantities $\xi$. Specifically, the behavior of these random quantities is fully determined by the knowledge of the family of continuous distributions:
\begin{equation}\label{DF}
dp_{\xi}(x|\theta)=\rho_{\xi}(x|\theta)dx,
\end{equation}
where the nonnegative function $\rho_{\mathbf{\xi}}(x|\theta)$ is the probability density, while $dx$ denotes the ordinary volume element (Lebesgue measure of the $n$-dimensional real space $\mathbb{R}^{n}$). Denoting by $\mathcal{S}\subset\mathcal{R}_{x}$, the integral:
\begin{equation}
p_{\mathbf{\xi}}(\mathcal{S}|\theta)=\int_{x\in\mathcal{S}} dp_{\mathbf{\xi}}(x|\theta)
\end{equation}
provides the probability that the complete set of random quantities $\xi$ takes any value $x\in\mathcal{S}$ under the external conditions $\mathfrak{E}$ with coordinates $\theta=(\theta^{1},\theta^{2},\ldots,\theta^{m})$. Considering the diffeomorphisms  $\xi:\mathcal{M}\rightarrow\mathcal{R}_{x}$ and $\theta:\mathcal{P}\rightarrow\mathcal{R}_{\theta}$, it is evident that continuous distributions family (\ref{DF}) provides a coordinate representation for the abstract distributions family of random events $dp(\epsilon|\mathfrak{E})$:
\begin{equation}
dp(\epsilon|\mathfrak{E})\equiv dp_{\xi}(x|\theta).
\end{equation}
Here, each point $\theta\in\mathcal{R}_{\theta}$ is associated with only one member of the distributions family (\ref{DF}). The set of continuous variables $\theta=(\theta^{1},\theta^{2},\ldots,\theta^{m})$ arise here as \emph{control parameters} of distributions family (\ref{DF}) because of the same ones parameterize the shape of these distributions. Consequently, the statistical manifold of external conditions $\mathcal{P}$ can be also referred to as the \emph{statistical manifold of distribution functions}. The concreted mathematical form of the distributions family (\ref{DF}) can be reconstructed from the experiment using the methods of statistical inference \cite{Lehmann}. The analysis of such statistical methods is outside the interest of the present work. On the contrary, \emph{the main interest here concerns to the information about the abstract statistical manifolds $\mathcal{M}$ and $\mathcal{P}$ that is obtained from the knowledge of the family of continuous distributions} (\ref{DF}).

\begin{example}\label{csm.ex}
In statistical mechanics, Boltzmann-Gibbs distributions:
\begin{equation}\label{BG}
dp_{\xi}(U,O|\beta,w)=\frac{1}{Z(\beta,w)}\exp\left[-\beta(U+wO)\right]\Omega(U,O)dUdO
\end{equation}
are commonly employed to describe the macroscopic behavior of an open classical system in thermodynamic equilibrium \cite{Reichl}, with $\Omega(U,O)$ and $Z(\beta,w)$ being the so-called \emph{states density} and the \emph{partition function}, respectively. An elementary event $\epsilon$ here is that the open system is found in a given macroscopic state. Experimentally, the realization of a given macroscopic state $\epsilon$ can be parameterized by certain set of macroscopic observables $\xi\sim(U,O)$ with a direct mechanical interpretation, such as the internal energy $U$ and the generalized displacements $O=(V,\mathbf{M}, \mathcal{M}, \ldots)$, in particular, the volume $V$, the total angular momentum $\mathbf{M}$, the magnetization $\mathcal{M}$, etc. The external conditions $\mathfrak{E}$ are parameterized by control parameters $\theta=(\beta,w)$ with an intrinsic statistical significance, such as the environmental inverse temperature $\beta=1/kT$ ($k$ is Boltzmann's constant) and the external thermodynamic forces $w=(p,-\mathbf{\omega},-\mathbf{H},\ldots)$, in particular, the external pressure $p$, the rotation frequency $\mathbf{\omega}$, the external magnetic field $\mathbf{H}$, etc.
\end{example}

\begin{example}\label{qm.ex}
Quantum mechanic provides other examples of continuous distributions. A particular case is the spatial distribution of a $N$-body non-relativistic quantum system:
\begin{equation}\label{QM.dist}
dp_{\mathbf{\xi}}(\mathbf{x}|a)=\left|\Psi(\mathbf{x};a)\right|^{2}d^{3N}\mathbf{x}.
\end{equation}
Here, the random elementary event $\epsilon$ is that the quantum system (a set of $N$ non-relativistic microparticles) is found in the positions $\mathbf{x}=(\mathbf{x}_{1},\mathbf{x}_{2},\ldots,\mathbf{x}_{N})$ of the $N$-body configuration space $\mathfrak{P}^{N}$. Experimentally, one needs to adopt certain reference frame to provide a quantitative parameterization of the physical space $\mathfrak{P}$, as well as the consideration of a given coordinate system (cartesian coordinates, polar coordinates, spherical coordinates, etc.). Here, $\Psi(\mathbf{x};a)$ denotes the so-called \emph{wave function}:
\begin{equation}\label{Wave}
\Psi(\mathbf{x};a)=\sum_{k}a_{k}\Psi_{k}(\mathbf{x}),
\end{equation}
which is expanded using certain basis $\left\{\Psi_{k}(\mathbf{x})\right\}$ of the Hilbert space $\mathcal{H}$. The external conditions $\mathfrak{E}$ correspond to the so-called \emph{preparations} of a quantum state \cite{APeres}, whose control parameters $\theta\sim a=(a_{k})$ are the \emph{wave amplitudes}. It is worth remarking that although the wave amplitudes represent a set of complex numbers, each complex number can be represented by an array of two real numbers, so that, the abstract manifold $\mathcal{P}$ can also be represented by a subset of a certain $m$-dimensional real space $\mathbb{R}^{m}$.
\end{example}

\begin{figure}
\begin{center}
\includegraphics[width=3.5in]{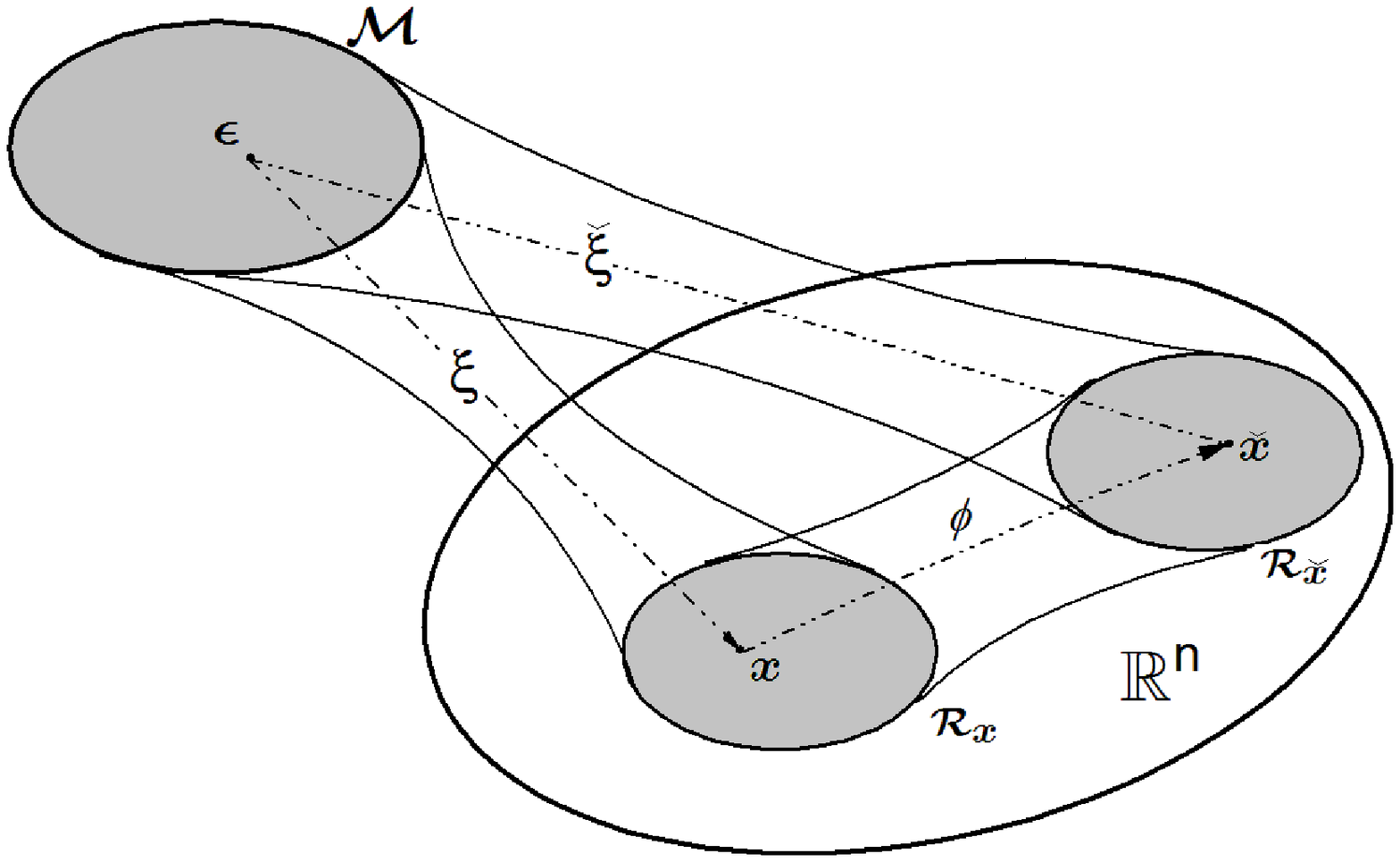}
\end{center}
\caption{The complete sets random quantities $\xi$ and $\check{\xi}$ represent two diffeomorphisms of the abstract statistical manifold $\mathcal{M}$ on the real subsets $\mathcal{R}_{x}$ and $\mathcal{R}_{\check{x}}\in\mathbb{R}^{n}$. The real subsets  $\mathcal{R}_{x}$ and $\mathcal{R}_{\check{x}}$ constitute \emph{coordinate representations} of abstract statistical manifold $\mathcal{M}$. Here, an (elementary) event $\epsilon$ is applied on the points $x=\xi(\epsilon)$ and $\check{x}=\check{\xi}(\epsilon)$. Additionally, it is illustrated the map $\phi:\mathcal{R}_{x}\rightarrow \mathcal{R}_{x}$, which represents a \emph{coordinate change} $\check{x}=\phi(x)$ between the coordinate representations $\mathcal{R}_{x}$ and $\mathcal{R}_{\check{x}}$ of the statistical manifold $\mathcal{M}$. This coordinate transformation also establishes a map among the complete sets of random quantities, $\check{\xi}=\phi(\xi)$.}\label{scheme.eps}
\end{figure}

\subsection{Coordinate changes and diffeomorphic distributions}\label{coord.tansf}

Quantitative characterization of the abstract statistical manifolds $\mathcal{M}$ and $\mathcal{P}$ demands to consider some coordinate representations of them. In particular, one needs to choose a complete set of random quantities $\xi$ to parameterize the occurrence of a given event $\epsilon$.  This choice can always be performed in multiple ways. Let us consider two different complete sets of random quantities $\xi$ and $\check{\xi}$, and let us denote by $x$ and $\check{x}$ certain admissible values of these random quantities (coordinates points of the real subsets $\mathcal{R}_{x}$ and $\mathcal{R}_{\check{x}}$, respectively). Since $\xi$ and $\check{\xi}$ represent two diffeomorphisms of the statistical manifold $\mathcal{M}$, one can introduce the map $\phi=\check{\xi} o \xi^{-1}$. This map defines a diffeomorphism $\phi:\mathcal{R}_{x}\rightarrow \mathcal{R}_{\check{x}}$ between the real subsets $\mathcal{R}_{x}$ and $\mathcal{R}_{\check{x}}$. The previous reasonings implies that two complete sets of random quantities $\check{\xi}$ and $\xi$ are related by a certain map $\phi$, $\check{\xi}\equiv \phi(\xi)$. Expressed in terms of the coordinates $x$ and $\check{x}$, the map $\check{x}=\phi(x)$ will be referred to as a \emph{coordinate change} (or re-parametrization), which is schematically illustrated in figure \ref{scheme.eps}. It is easy to realize that the identity:
\begin{equation}
dp_{\check{\xi}}(\check{x}|\theta)=dp_{\xi}(x|\theta)
\end{equation}
takes place because of the points $x$ and $\check{x}$ correspond to a same elementary event $\epsilon\in\mathcal{M}$, as well as the points $x+dx$ and $\check{x}+d\check{x}$ correspond to an elementary event $\epsilon'\in\mathcal{M}$ that is infinitely close to the event $\epsilon$. Thus, one obtains the following transformation rule for the probability density:
\begin{equation}\label{transf.rule}
\rho_{\check{\xi}}(\check{x}|\theta)=\rho_{\xi}(x|\theta)\left|\frac{\partial \check{x}}{\partial x}\right|^{-1},
\end{equation}
with $\left|\partial \check{x}/\partial x\right|$ being the Jacobian of the coordinate change $\check{x}=\phi(x)$. Analogously, it is possible to consider coordinate change $\check{\theta}=\varphi(\theta)$ for the statistical manifold $\mathcal{P}$, $\varphi:\mathcal{R}_{\theta}\rightarrow\mathcal{R}_{\check{\theta}}$. Let us consider two simple illustration examples.

\begin{example}\label{BMexample}
The family of gaussian distributions:
\begin{equation}\label{ex.twogauss}
dp_{\check{\xi}}(\check{x},\check{y}|\theta)=\frac{1}{2\pi \theta^{2}}\exp\left[-(\check{x}^{2}+\check{y}^{2})/2\theta^{2}\right]d\check{x}d\check{y},
\end{equation}
can be obtained from the distributions family:
\begin{equation}\label{ex.exp}
dp_{\xi}(x,y|\theta)=\frac{1}{2\pi \theta^{2}}\exp\left[-x^{2}/2\theta^{2}\right]xdxdy
\end{equation}
considering the coordinate change $(\check{x},\check{y})=\phi(x,y)$ defined by:
\begin{equation}\label{tt.ex1}
\check{x}=x\cos y\mbox{ and }\check{y}=x\sin y.
\end{equation}
Noteworthy that the real variables $(\check{x},\check{y})\in\mathcal{R}_{\check{x}}\equiv\mathbb{R}^{2}$ and the control parameter $\theta\in\mathcal{R}_{\theta}\equiv\mathbb{R}^{+}$ ($\mathbb{R}^{+}$ is the subset of real positive numbers). On the other hand, $(x,y)\in\mathcal{R}_{x}$ whenever $0\leq x<+\infty$ and $0\leq y\leq2\pi$. As additional restrictions, it is necessary to identify the points $(x,0)$ and $(x,2\pi)$, as well as every point on the segment $(0,y)$.
\end{example}

\begin{example}
Gaussian distribution with mean $\mu$ and variance $\sigma$:
\begin{equation}\label{gauss.introd}
dp_{\xi}(x|\mu,\sigma)=\frac{1}{\sqrt{2\pi}\sigma}\exp\left[-(x-\mu)^{2}/2\sigma^{2}\right]dx
\end{equation}
can be re-parameterized as follows:
\begin{equation}\label{ass.gauss}
dp_{\xi}(x|\check{\theta}^{1},\check{\theta}^{2})=\frac{1}{z(\check{\theta}^{1},\check{\theta}^{2})}\exp
\left[-\check{\theta}^{1}x-\check{\theta}^{2}x^{2}\right]dx
\end{equation}
considering the coordinate change $(\check{\theta}^{1},\check{\theta}^{2})=\varphi(\mu,\sigma)$ defined by:
\begin{equation}\label{tt.ex2}
\check{\theta}^{1}=-\mu/\sigma^{2}\mbox{ and }\check{\theta}^{2}=1/2\sigma^{2}.
\end{equation}
Here, the normalization factor $z(\check{\theta}^{1},\check{\theta}^{2})$ is given by:
\begin{equation}
z(\check{\theta}^{1},\check{\theta}^{2})=\sqrt{\pi/\check{\theta}^{1}}\exp\left[(\check{\theta}^{1})^{2}/4\check{\theta}^{2}\right].
\end{equation}
Moreover, $x\in\mathcal{R}_{x}\equiv \mathbb{R}$, the real subset $\mathcal{R}_{\theta}$ with coordinates $\theta=(\mu,\sigma)$ is the semi-plane of $\mathbb{R}^{2}$ with $\sigma>0$, while the subset $\mathcal{R}_{\check{\theta}}$ with coordinates $\check{\theta}=(\check{\theta}^{1},\check{\theta}^{2})$ is also a semi-plane of $\mathbb{R}^{2}$ with $\check{\theta}^{2}>0$.
\end{example}

The possibility to consider different coordinates representations for the abstract statistical manifolds $\mathcal{M}$ and $\mathcal{P}$ introduces a great flexibility into the statistical analysis. In fact, some coordinate representations are more suitable than others for some practical purposes. For example, coordinate change (\ref{tt.ex1}) is a key assumption to demonstrate the improper integral:
\begin{equation}
\int^{+\infty}_{-\infty}\exp(-x^{2})dx=\sqrt{\pi},
\end{equation}
which is employed to derive the normalization constant of Gaussian distributions. This coordinate change is the basis of Box-Muller transform to generate gaussian pseudo-random numbers \cite{Box-Muller}:
\begin{equation}
x=\mu+\sigma\sqrt{-2\log(\zeta_{1})}\cos(2\pi\zeta_{2}).
\end{equation}
Here, $\zeta_{1}$ and $\zeta_{2}$ are two independent pseudo-random numbers that are uniformly distributed in the interval $(0,1]$. On the other hand, the coordinate change (\ref{tt.ex2}) clearly evidences that the Gaussian distributions (\ref{gauss.introd}) is a member of \emph{exponential family} (notice that Boltzmann-Gibbs distributions (\ref{BG}) also belong to this family). According to Pitman-Koopman theorem \cite{Koopman}, only the exponential family guarantees the existence of \emph{sufficient estimators}. Additionally, the resulting representation (\ref{gauss.introd}) exhibits a more convenient mathematical form to calculate the so-called \emph{Fisher's information matrix} \cite{Fisher} (see Eq.(\ref{Fisher}) below), which allows to establish the Cramer-Rao lower bound of \emph{unbiased estimators} \cite{Rao}. The previous examples motivate the introduction of the notion of \emph{diffeomorphic distributions}.

\begin{definition}
\textbf{Diffeomorphic distributions} are those continuous distributions whose associated complete sets of random quantities $\xi$ and $\check{\xi}$ are related by means of a certain differentiable map $\phi$, $\check{\xi}=\phi(\xi)$; and hence, they can be regarded as two different coordinate representations of a same continuous distribution $dp(\epsilon|\mathfrak{E})$ defined on the abstract statistical manifolds $\mathcal{M}$ and $\mathcal{P}$.
\end{definition}

Two diffeomorphic distributions are fully equivalent from the viewpoint of their geometrical properties. Examples of diffeomorphic distributions are the distributions families (\ref{ex.twogauss}) and (\ref{ex.exp}), as well as distributions families (\ref{gauss.introd}) and (\ref{ass.gauss}). Interestingly, the notion of diffeomorphic distributions comprises some continuous distributions families with a very different statistical behavior.

\begin{example}\label{diffeormorphic}
At first glance, the statistical features of Gaussian distributions (\ref{gauss.introd}) significantly differ from the ones of \emph{Cauchy distributions}:
\begin{equation}\label{lorentzian}
dp_{\check{\xi}}(\check{x}|\nu,\gamma)=\frac{1}{\pi}\frac{\gamma d\check{x}}{\gamma^{2}+(\check{x}-\nu)^{2}}.
\end{equation}
For example, the mean, the variance and every positive integer $n$-th moment of a random quantity $\xi$ that obeys Gaussian distributions (\ref{gauss.introd}) do exit and they are finite, while the ones associated with a random quantity $\check{\xi}$ that obeys Cauchy distributions (\ref{lorentzian}) do not exist (or diverges). However, it is possible to verify that the coordinate change $\phi:\mathcal{R}_{x}\rightarrow \mathcal{R}_{\check{x}}$ defined by:
\begin{equation}\label{map1}
\check{x}=\phi(x|\mu,\sigma;\nu,\gamma)=\nu+\gamma\tan\left[\frac{\pi}{2} \mathrm{erf}\left(\frac{x-\mu}{\sqrt{2}\sigma}\right)\right]
\end{equation}
establishes a diffeomorphism between the distributions families (\ref{gauss.introd}) and (\ref{lorentzian}). Here, $\mathrm{erf}(s)$ denotes the error function:
\begin{equation}
\mathrm{erf}(s)=\frac{2}{\sqrt{\pi}}\int^{s}_{0}e^{-\tau^{2}}d\tau.
\end{equation}
Consequently, distributions families (\ref{gauss.introd}) and (\ref{lorentzian}) are diffeomorphic distributions.
\end{example}
\begin{remark}
Continuous distributions families whose abstract statistical manifolds $\mathcal{M}$ are diffeomorphic to the one-dimensional real space $\mathbb{R}$ are diffeomorphic distributions.
\end{remark}
\begin{proof}
Let us consider two different distributions families $dp_{\xi}(x|\theta)$ and $dp_{\check{\xi}}(\check{x}|\check{\theta})$ of this class of distributions. Let us now consider their cumulative distribution functions:
\begin{equation}\label{cdf}
p_{\xi}(y|\theta)=\int^{y}_{x_{min}}\rho_{\xi}(x|\theta)dx\mbox{ and }p_{\check{\xi}}(\check{y}|\check{\theta})=\int^{\check{y}}_{\check{x}_{min}}\rho_{\check{\xi}}(\check{x}|\check{\theta})d\check{x}
\end{equation}
with $x_{min}$ and $\check{x}_{min}$ being the minimum admissible values of the random quantities $\xi$ and $\check{\xi}$. By definition, the cumulative distribution functions (\ref{cdf}) are absolutely continuous and differentiable, so that, these functions represent diffeomorphisms of the real subsets $\mathcal{R}_{x}$ and $\mathcal{R}_{\check{x}}$ on the interval $(0,1)\subset \mathbb{R}$. The coordinate change $\check{x}=\phi(x|\theta,\check{\theta})$ defined by:
\begin{equation}\label{map2}
\check{x}=p^{-1}_{\check{\xi}}\left[p_{\xi}(x|\theta)|\check{\theta}\right],
\end{equation}
represents a diffeomorphism $\phi:\mathcal{R}_{x}\rightarrow \mathcal{R}_{\check{x}}$ between the real one-dimensional subsets $\mathcal{R}_{x}$ and $\mathcal{R}_{\check{x}}$.
\end{proof}

The coordinate change (\ref{map1}) is a particular case of the map (\ref{map2}). Noteworthy that this type of coordinate changes is much general than the coordinate change (\ref{tt.ex1}) because of it also involves the control parameters of the associated distributions families. In computational applications, the map (\ref{map2}) is the basis of the so-called \emph{inverse transformation method} for nonuniform pseudo-random number sampling \cite{Devroye}.

\subsection{Relative statistical properties}\label{relative}

In principle, family of continuous distributions (\ref{DF}) contains all the necessary information about the distribution function $dp(\epsilon|\mathfrak{E})$ defined on the abstract statistical manifolds $\mathcal{M}$ and $\mathcal{P}$. However, this family also provides information that is \emph{relative} to their concrete coordinate representations $\mathcal{R}_{x}$ and $\mathcal{R}_{\theta}$. An obvious relative property is the mathematical form of these distributions. According to transformation rule (\ref{transf.rule}), the local values of the probability density are generally modified by a coordinate change. The relative character of some properties of continuous distributions put in evidence the restricted applicability of certain statistical notions.

Strictly speaking, the probability that a continuous random quantity $\xi$ takes a given value $x$ is zero. Nevertheless, an usual question in many practical applications is to find out \emph{the most likely value} $\bar{x}$ of a random quantity $\xi$. A criterium widely employed is to identify the point $\bar{x}$ where the associated probability density $\rho_{\xi}(x|\theta)$ reaches a global maximum. As expected, the point $\bar{x}\in\mathcal{R}_{x}$ univocally parameterizes the occurrence of a random event $\bar{\epsilon}_{x}\in\mathcal{M}$, so that, one may be tempted to regard $\bar{\epsilon}_{x}$ as \emph{the most likely event}. However, a re-examination of this argument evidences its restricted applicability. It is easy to realize that the most likely elementary event associated with this criterium crucially depends on the coordinate representation $\mathcal{R}_{x}$. A simple illustration of this fact is shown in figure \ref{entropy.eps}, where the point global maximum of a gaussian distribution turns a point of global minimum of a two-peaks gaussian distribution using an appropriate coordinate change of the form (\ref{map2}).

The notion of diffeomorphic distributions reveals an inconsistence associated with the concept of \emph{information entropy for continuous distributions}. Conventionally, the information entropy is introduced as \emph{a global measure of unpredictability (or uncertainty) of a random quantity} $\xi$. For random quantities $\xi$ that exhibit a discrete spectrum of admissible values $\left\{x_{k}\right\}$, the information entropy is defined as follows:
\begin{equation}\label{inf.entropy}
S\left[\xi|\theta\right]=-\sum_{k}p_{\xi}(x_{k}|\theta)\log p_{\xi}(x_{k}|\theta),
\end{equation}
with $p_{\xi}(x_{k}|\theta)$ being the probability of the $k$-th admissible values of the set of random quantities $\xi$. The usual extension of this notion in the framework of continuous distributions is given by the following integral (in Lebesgue's sense):
\begin{equation}\label{de.Lebesgue}
S_{d}(\xi|\theta)=-\int_{x\in\mathcal{R}_{x}}\log\left[\rho_{\xi}(x|\theta)\right]\rho_{\xi}(x|\theta)dx,
\end{equation}
which is referred to as \emph{differential entropy} in the literature \cite{Dekking}. According to the transformation rule (\ref{transf.rule}), the information entropy (\ref{de.Lebesgue}) provides different values for those random quantities $\xi$ and $\check{\xi}$ that are related by a diffeomorphism $\check{\xi}=\phi(\xi)$:
\begin{equation}
S_{d}(\check{\xi}|\theta)-S_{d}(\xi|\theta)=\left\langle\log\left| \partial \check{x}/\partial x\right|\right\rangle
=\int_{x\in\mathcal{R}_{x}}\log\left|\partial \check{x}/\partial x\right|\rho_{\xi}(x|\theta)dx.
\end{equation}
Consequently, differential entropy (\ref{de.Lebesgue}) is a relative statistical property. However, this fact contrast with the notion that diffeomorphic distributions actually represent different coordinate representations of a same abstract distribution. According to this interpretation, \emph{diffeomorphic distributions should exhibit the same value of information entropy}! For example, continuous distributions show in figure \ref{entropy.eps} are diffeomorphic distributions, and hence, they should exhibit the same amount of information entropy.

The lack of invariance of differential entropy (\ref{de.Lebesgue}) was emphasized by Jaynes \cite{Jaynes}. This author proposed to overcome this inconsistence introducing other positive measure defined on the statistical manifold $\mathcal{M}$:
\begin{equation}\label{mu}
d\mu(x)=\varrho(x)dx,
\end{equation}
and redefining (\ref{de.Lebesgue}) as follows:
\begin{equation}\label{de.relative}
S^{\mu}_{d}(\xi|\theta)=-\int_{x\in\mathcal{R}_{x}}\left[dp_{\xi}(x|\theta)/d\mu(x)\right]\log\left[dp_{\xi}(x|\theta)/d\mu(x)\right]d\mu(x).
\end{equation}
At first glance, the \emph{relative entropy} (\ref{de.relative}) is similar to the \emph{Kullback-Leibler divergence} \cite{Kullback}, but its meaning is different, overall, because of the measure (\ref{mu}) is not necessarily a probability distribution. However, I think that the ansatz (\ref{de.relative}) is not a suitable extension for the information entropy (\ref{inf.entropy}). For example, equation (\ref{inf.entropy}) only depends on the discrete distribution, while definition (\ref{de.relative}) involves a second independent measure $d\mu(x)$. A natural question here is how to introduce the measure (\ref{mu}) when no other information is available, except the continuous distribution (\ref{DF}).

\begin{figure}
\begin{center}
  \includegraphics[width=4.5in]{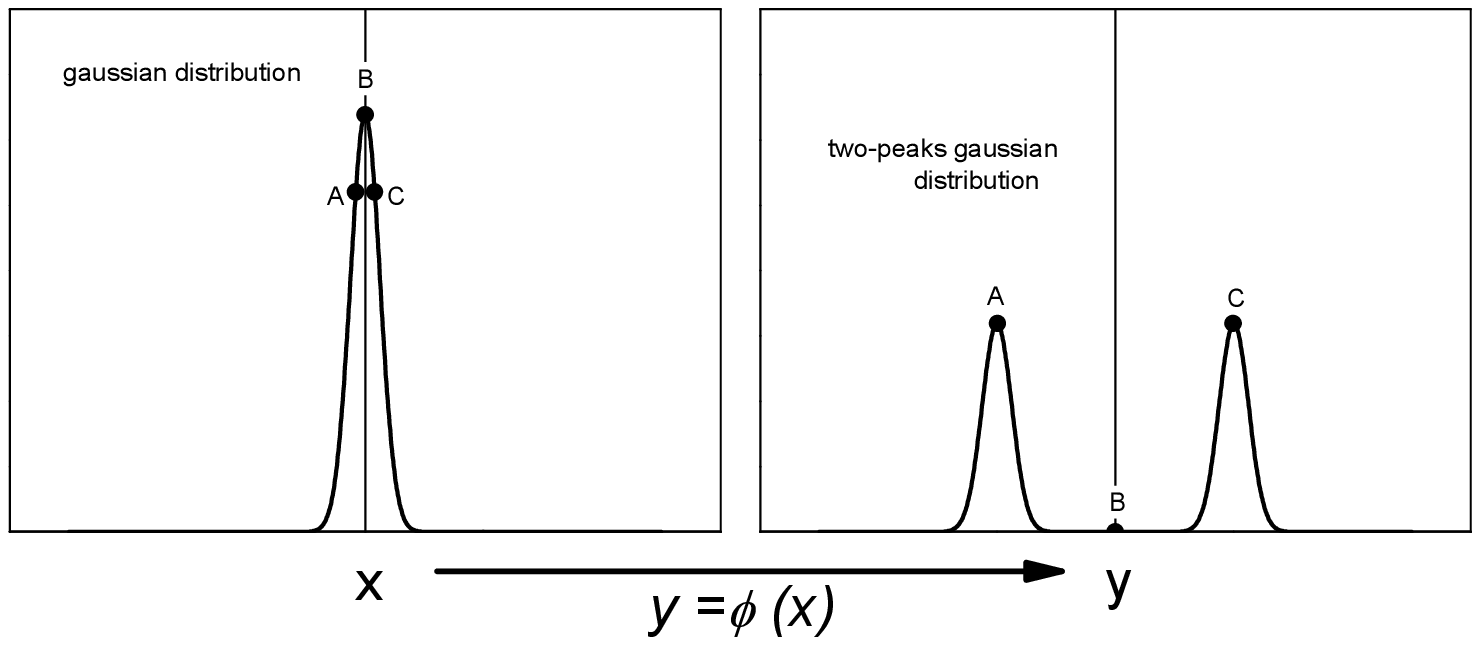}
\end{center}
\caption{According to definition (\ref{de.Lebesgue}), the random quantity $\xi$ that obeys a two-peaks gaussian distribution with well-separated peaks of width $\sigma$ exhibits an amount of information entropy $\delta S\simeq2\log 2$ larger than the random quantity $\check{\xi}$ that obeys a gaussian distribution with only one peak and the same width $\sigma$. These distributions are \emph{diffeormorphic distributions} because of they are related by a coordinate change $y=\phi(x)$ of the form (\ref{map2}). Although counterintuitive, these distributions should exhibit the same amount of information entropy. The points \textbf{A}, \textbf{B} and \textbf{C} of each distribution are related by the map $y=\phi(x)$. As expected, global maxima of a continuous distribution can be modified in a radical way under a coordinate change.}\label{entropy.eps}
\end{figure}

\subsection{Riemannian geometries of the statistical manifolds $\mathcal{M}$ and $\mathcal{P}$}\label{R.structures}

Statistical theory should enable us to characterize those \emph{absolute (or intrinsic) properties} of the abstract statistical manifolds $\mathcal{M}$ and $\mathcal{P}$ without reference to particular coordinate representations, that is, to perform a \emph{coordinate-free treatment}\footnote{A coordinate-free treatment of a scientific theory develops its concepts on any form of manifold without reference to any particular coordinate system. Coordinate-free treatments generally allow for simpler systems of equations and inherently constrain certain types of inconsistency, allowing greater mathematical elegance at the cost of some abstraction from the detailed formulae needed to evaluate these equations within a particular system of coordinates.}. This goal can be achieved using the mathematical apparatus of Riemannian geometry \cite{Berger}, in particular, introducing a \emph{Riemannian structure} for the statistical manifolds $\mathcal{M}$ and $\mathcal{P}$.

As pioneering suggested by Rao \cite{Rao}, the statistical manifold $\mathcal{P}$ can be endowed of a \emph{Riemannian structure} using the distance notion:
\begin{equation}\label{inf.dist}
ds^{2}=\mathfrak{g}_{\alpha\beta}(\theta)d\theta^{\alpha}d\theta^{\beta},
\end{equation}
where the metric tensor $\mathfrak{g}_{\alpha\beta}(\theta)$ is the \emph{Fisher's inference matrix} \cite{Fisher}:
\begin{equation}\label{Fisher}
\mathfrak{g}_{\alpha\beta}(\theta)=\int_{\mathcal{M}}\frac{\partial\log\rho_{\mathbf{\xi}}(x|\theta)}{\partial\theta^{\alpha}}
\frac{\partial\log\rho_{\mathbf{\xi}}(x|\theta)}{\partial\theta^{\beta}}
dp_{\mathbf{\xi}}(x|\theta).
\end{equation}
The distance notion of (\ref{inf.dist}) characterizes the statistical separation between two different members of the distributions family (\ref{DF}), that is, a global measure about modification of behavior of random quantities $\xi$ under two external conditions $\mathfrak{E}$ and $\mathfrak{E}'\in\mathcal{P}$ that are infinitely close. As discussed elsewhere \cite{Amari}, this distance is a measure of the \emph{distinguishing probability} of these distributions during a procedure of statistical inference. By its statistical significance, this type of statistical geometry could be referred to as Riemannian geometry of inference theory, or more briefly, \emph{inference geometry}. However, this approach is now widely known as \emph{information geometry} in the literature \cite{Amari}.

Alternatively, the statistical manifold $\mathcal{M}$ can be also endowed of a Riemannian structure using the distance notion:
\begin{equation}\label{fluct.dist}
ds^{2}=g_{ij}(x|\theta)dx^{i}dx^{j},
\end{equation}
where the metric tensor $g_{ij}=g_{ij}(x|\theta)$ should be obtained from the probability density $\rho_{\mathbf{\xi}}=\rho_{\mathbf{\xi}}(x|\theta)$ as the solution of a set of covariant partial differential equations \cite{Vel.FG}:
\begin{equation}  \label{cov.equation}
g_{ij}=-\frac{\partial^{2}\log\rho_{\mathbf{\xi}}}{\partial x^{i}\partial x^{j}}+
\Gamma^{k}_{ij}\frac{\partial\log\rho_{\mathbf{\xi}}}{\partial x^{k}} +\frac{\partial\Gamma^{k}_{jk}}{\partial x^{i}}-\Gamma^{k}_{ij}\Gamma^{l}_{kl}.
\end{equation}
Here, $\Gamma^{k}_{ij}=\Gamma^{k}_{ij}(x|\theta)$ are the \emph{Levi-Civita affine connections} \cite{Berger} (see equation (\ref{Levi-Civita}) below). Equation (\ref{cov.equation}) represents a set of covariant partial differential equations of second-order with respect to the metric tensor $g_{ij}(x|\theta)$. Its covariant character can be demonstrated starting from the transformation rule of the metric tensor:
\begin{equation}\label{metric.tr}
\check{g}_{kl}(\check{x}|\theta)=\frac{\partial x^{i}}{\partial \check{x}^{k}}\frac{\partial x^{j}}{\partial \check{x}^{j}}g_{ij}(x|\theta)
\end{equation}
and the transformation rule of the probability density (\ref{transf.rule}). The distance notion (\ref{fluct.dist}) represents a statistical separation between two infinitely close random events $\epsilon$ and $\epsilon'\in\mathcal{M}$ under the same external conditions $\mathfrak{E}$. This second distance provides a measure about the \emph{relative occurrence probability} of these events. Due to its statistical relevance, this geometry can be referred to as Riemannian geometry of fluctuation theory, or more briefly, \emph{fluctuation geometry} \cite{Vel.FG}.

The above geometries establishes a direct relationship among the statistical properties of the distributions family (\ref{DF}) and the geometric features of the abstract statistical manifolds $\mathcal{M}$ and $\mathcal{P}$. Consequently, these approaches enable us to employ the powerful tools of Riemannian geometry for proving statistical results.

\begin{table}[tbp] \centering
\begin{tabular}{cc}
\hline\hline
\textbf{Notation} & \textbf{Meaning} \\ \hline
\multicolumn{1}{l}{%
\begin{tabular}{l}
$\mathcal{M}${\small \ and }$\mathcal{P}$%
\end{tabular}%
} & \multicolumn{1}{l}{%
\begin{tabular}{c}
{\small statistical manifolds of the random elementary events }$\epsilon $
\\
\multicolumn{1}{l}{\small and the external conditions $\mathfrak{E}$ respectively.}%
\end{tabular}%
} \\
\multicolumn{1}{l}{$%
\begin{array}{l}
x{\small =}\left( x^{1},\ldots x^{i}\ldots x^{n}\right) , \\
\check{x}{\small =}\left( \check{x}^{1},\ldots \check{x}^{i}\ldots \check{x}%
^{n}\right)
\end{array}%
$} & \multicolumn{1}{l}{%
\begin{tabular}{l}
{\small general coordinates of the manifold }$\mathcal{M}$%
\end{tabular}%
} \\
\multicolumn{1}{l}{%
\begin{tabular}{l}
$\mathcal{R}_{x}${\small \ and }$\mathcal{R}_{\check{x}}$%
\end{tabular}%
} & \multicolumn{1}{l}{%
\begin{tabular}{l}
{\small coordinates representations of the manifold }$\mathcal{M}$%
\end{tabular}%
} \\
\multicolumn{1}{l}{%
\begin{tabular}{l}
$\phi :\mathcal{R}_{x}\rightarrow \mathcal{R}_{\check{x}}$%
\end{tabular}%
} & \multicolumn{1}{l}{%
\begin{tabular}{c}
{\small coordinate change of }$\mathcal{M}$%
\end{tabular}%
} \\
\multicolumn{1}{l}{%
\begin{tabular}{l}
$\left( \ell ,q\right) ${\small \ with }$q=\left( q^{1},q^{2},\ldots
q^{n-1}\right) $ \\
{\small and }$\mathcal{R}_{\rho }$%
\end{tabular}%
} & \multicolumn{1}{l}{%
\begin{tabular}{l}
{\small radial and angular coordinates associated with} \\
{\small the spherical representation of }$\mathcal{M}$%
\end{tabular}%
} \\
\multicolumn{1}{l}{%
\begin{tabular}{l}
$\theta {\small =}\left( \theta ^{1},\ldots \theta ^{\alpha },\ldots \theta
^{m}\right) $%
\end{tabular}%
} & \multicolumn{1}{l}{%
\begin{tabular}{l}
{\small coordinates (control parameters) of the manifold }$\mathcal{P}$%
\end{tabular}%
} \\
\multicolumn{1}{l}{%
\begin{tabular}{l}
$\mathcal{R}_{\theta }$%
\end{tabular}%
} & \multicolumn{1}{l}{%
\begin{tabular}{l}
{\small coordinate representation of the manifold }$\mathcal{P}${\small \ }%
\end{tabular}%
} \\
\multicolumn{1}{l}{%
\begin{tabular}{l}
$g_{ij}\left( x|\theta \right) ${\small \ and }$\mathfrak{g}_{\alpha \beta
}\left( \theta \right) $%
\end{tabular}%
} & \multicolumn{1}{l}{%
\begin{tabular}{l}
{\small metric tensors of the statistical manifolds }$\mathcal{M}$ and $%
\mathcal{P}$%
\end{tabular}%
} \\
\multicolumn{1}{l}{%
\begin{tabular}{l}
$dx${\small \ and }$d\mu \left( x|\theta \right) $%
\end{tabular}%
} & \multicolumn{1}{l}{%
\begin{tabular}{l}
{\small ordinary volume element (Lebesgue measure) and invariant } \\
{\small volume element of the manifold }$\mathcal{M}$%
\end{tabular}%
} \\
\multicolumn{1}{l}{%
\begin{tabular}{l}
$\rho \left( x|\theta \right) ${\small \ and }$\omega \left( x|\theta
\right) $%
\end{tabular}%
} & \multicolumn{1}{l}{%
\begin{tabular}{l}
{\small probability density and probability weight}%
\end{tabular}%
} \\
\multicolumn{1}{l}{%
\begin{tabular}{l}
$\mathcal{S}\left( x|\theta \right) ${\small \ and }$\mathfrak{I}\left(
x|\theta \right) $%
\end{tabular}%
} & \multicolumn{1}{l}{%
\begin{tabular}{l}
{\small information potential and local information content}%
\end{tabular}%
} \\
\multicolumn{1}{l}{%
\begin{tabular}{l}
$\ell _{\theta }\left( x,\bar{x}\right) $%
\end{tabular}%
} & \multicolumn{1}{l}{%
\begin{tabular}{l}
{\small separation distance between two points of the manifold }$\mathcal{M}$
\\
with {\small coordinates }$x$ and $\bar{x}$%
\end{tabular}%
} \\
\multicolumn{1}{l}{%
\begin{tabular}{l}
$R_{ijkl}\left( x|\theta \right) $ and $R(x|\theta )$%
\end{tabular}%
} & \multicolumn{1}{l}{%
\begin{tabular}{l}
{\small fourth-rank curvature tensor and curvature scalar of $\mathcal{M}$}
\end{tabular}%
} \\ \hline\hline
\end{tabular}%
\caption{Summary of most usual notations and symbols employed along this work. Occasionally, other symbols have been
 employed, overall, in examples and applications. Their usage should be clear from the context.}
\label{Notation}
\end{table}

\subsection{About the mathematical notations and conventions}

A summary of most usual notations and symbols employed in this work are shown in table \ref{Notation}. These notations are slightly different than the ones considered in precedent works \cite{Vel.GEFT,Vel.Book} and \cite{Vel.CSM}-\cite{Vel.FG}, but they are closer to the standard ones employed in mathematical statistics and differential geometry. Einstein summation convention of repeated indexes has been also assumed. Hereinafter, all mathematical relations are expressed in the same mathematical appearance without mattering the coordinate representation of the statistical manifolds $\mathcal{M}$ and $\mathcal{P}$, that is, a coordinate-free treatment will be adopted. This goal is achieved rephrasing the statistical description using tensorial quantities of Riemannian geometry.

Relations involving tensorial quantities can be divided into two categories: (i) \emph{tensorial relations}, which describe how are related different tensorial quantities, and (ii) the \emph{covariant transformation rules}, which express how the components of a certain tensorial quantity are modified under a coordinate change $\phi:\mathcal{R}_{x}\rightarrow \mathcal{R}_{\check{x}}$. Equation (\ref{cov.equation}) is a particular example of tensorial relation, that establishes a constraint between the metric tensor $g_{ij}(x|\theta)$ and the probability density $\rho_{\xi}(x|\theta)$. Noteworthy that these relations presuppose that all tensorial quantities are expressed using the same coordinates representations of the manifolds $\mathcal{M}$ and $\mathcal{P}$. Examples of covariant transformation rules are the ones considered in equations (\ref{transf.rule}) and (\ref{metric.tr}). Let us denote by $a^{j_{1}j_{2}\ldots j_{q}}_{i_{1}i_{2}\ldots i_{p}}(x|\theta)$ and $\check{a}^{l_{1}l_{2}\ldots l_{q}}_{k_{1}k_{2}\ldots k_{p}}(\check{x}|\theta)$ the components of a certain tensor in the coordinate representations $\mathcal{R}_{x}$ and $\mathcal{R}_{\check{x}}$ of the manifold $\mathcal{M}$, respectively. Thus, the transformation rule of a tensorial entity of weight $W$ and rank $R=(p+q)$ ($p$-times covariant and $q$-times contravariant) reads as follows:
\begin{equation}\label{gen.tr}
\check{a}^{l_{1}l_{2}\ldots l_{q}}_{k_{1}k_{2}\ldots k_{p}}(\check{x}|\theta)=a^{j_{1}j_{2}\ldots j_{q}}_{i_{1}i_{2}\ldots i_{p}}(x|\theta)\left|\frac{\partial \check{x}}{\partial x}\right|^{W}\frac{\partial x^{i_{1}}}{\partial \check{x}^{k_{1}}}\frac{\partial x^{i_{2}}}{\partial \check{x}^{k_{2}}}\ldots\frac{\partial x^{i_{p}}}{\partial \check{x}^{k_{p}}}\frac{\partial \check{x}^{l_{1}}}{\partial x^{j_{1}}}\frac{\partial \check{x}^{l_{2}}}{\partial x^{j_{2}}}\ldots \frac{\partial \check{x}^{l_{q}}}{\partial x^{j_{q}}}.
\end{equation}
Hereinafter, coordinate changes involving control parameters $\theta$ shall not be considered.

The notation of the family of continuous distributions (\ref{DF}) will be simplified as follows:
\begin{equation}\label{DFxi}
dp(x|\theta)=\rho(x|\theta)dx
\end{equation}
without specifying the complete set of random quantities $\xi$. Of course, each coordinate representation $\mathcal{R}_{x}$ of the manifold $\mathcal{M}$ is associated with a complete set of random quantities $\xi$, the map $\xi:\mathcal{M}\rightarrow \mathcal{R}_{x}$. However, this association will be omitted here to adopt the usual terminology and nomenclature employed in Riemannian geometry. Considering the general transformation rule (\ref{gen.tr}) for tensorial quantities, transformation rule for the probability density (\ref{transf.rule}) can be re-expressed as follows:
\begin{equation}
\check{\rho}(\check{x}|\theta)=\rho(x|\theta)\left|\frac{\partial \check{x}}{\partial x}\right|^{-1}.
\end{equation}
The probability density $\rho(x|\theta)$ is a tensor of rank $R=0$ and weight $W=-1$, which is usually referred to as a \emph{scalar density}.

Any function $a(\xi|\theta)$ of a random quantity $\xi$ also represents a random quantity. However, the notation $a(x|\theta)$ will be regarded as an \emph{ordinary function} defined on the manifolds $\mathcal{M}$ and $\mathcal{P}$, which is expressed using the coordinate representations $\mathcal{R}_{x}$ and $\mathcal{R}_{\theta}$. For example, the probability density $\rho(x|\theta)$, the metric tensor $g_{ij}(x|\theta)$ and the curvature tensor $R_{ijkl}(x|\theta)$ are example of functions (or tensorial quantities) defined on the abstract statistical manifolds $\mathcal{M}$ and $\mathcal{P}$. Moreover, the notation $\left\langle a(x|\theta)\right\rangle$, as usual, refers to the statistical expectation value obtained from the knowledge of the family continuous distribution:
\begin{equation}
\left\langle a(x|\theta)\right\rangle\equiv\int_{\mathcal{M}} a(x|\theta)dp(x|\theta).
\end{equation}

\subsection{Some results of fluctuation geometry}\label{review}

Riemannian structure of the statistical manifold $\mathcal{M}$ allows us to introduce the \textit{invariant volume element} $d\mu(x|\theta)$:
\begin{equation}  \label{inv.volume}
d\mu(x|\theta)=\sqrt{\left\vert g_{ij}(x|\theta)/2\pi
\right\vert }dx,
\end{equation}
which replaces the ordinary volume element $dx$ that is employed in equation (\ref{DFxi}). The notation $\left|T_{ij}\right|$ represents the determinant of a given tensor $T_{ij}$ of second-rank, while the factor $2\pi$ has been introduced for convenience. Additionally, one can define the \textit{probabilistic weight} \cite{Vel.FG}:
\begin{equation}\label{scalar.weight}
\omega(x|\theta)=\rho(x|\theta)\sqrt{|2\pi g^{ij}(x|\theta)|},
\end{equation}
which is a scalar function that arises as \emph{a local invariant measure of the probability}. Although the mathematical form of the probabilistic weight $\omega(x|\theta)$ depends on the coordinates representations of the statistical manifolds $\mathcal{M}$ and $\mathcal{P}$; the values of this function are the same in all coordinate representations. Using the above notions, the family of continuous distributions (\ref{DFxi}) can be rewritten as follows:
\begin{equation}\label{DF2}
dp(x|\theta)=\omega(x|\theta)d\mu(x|\theta),
\end{equation}
which is a form that explicitly exhibits the invariance of this family of distributions.

The notion of probability weight $\omega(x|\theta)$ allows us to overcome the inconsistencies commented in subsection \ref{relative}. For example, its scalar character enables an unambiguous definition for the most likely event $\bar{\epsilon}\in\mathcal{M}$, precisely, the event corresponding to the point $\bar{x}$ of global maximum of the probability weight $\omega(x|\theta)$. Additionally, the notion of information entropy for continuous distributions (\ref{de.Lebesgue}) can be extended as follows \cite{Vel.FG}:
\begin{equation}\label{SCDF}
\mathcal{S}_{d}\left[\omega|g,\mathcal{M}\right]=-\int_{\mathcal{M}}\omega(x|\theta)\log\omega(x|\theta)d\mu(x|\theta).
\end{equation}
The quantity (\ref{SCDF}) is a now global invariant measure that depends on the metric tensor $g_{ij}(x|\theta)$ of the manifold $\mathcal{M}$. Noteworthy that the sum over different discrete values $\sum_{k}$ in definition (\ref{inf.entropy}) turns now an integral $\int d\mu(x|\theta)$ over \emph{distinguishable events}. Here, the quantity $\mathfrak{I}(x|\theta)$:
\begin{equation}\label{information.content}
\mathfrak{I}(x|\theta)=-\log\omega(x|\theta)
\end{equation}
represents a \emph{local invariant measurement of the information content}. By definition, differential entropy (\ref{SCDF}) exhibits the same value for all diffeomorphic distributions. Readers can find further details about this measure in subsection 6.2 of Ref.\cite{Vel.FG}.

Introducing the \textit{information potential} $\mathcal{S}(x|\theta)$ as the negative of the information content (\ref{information.content}):
\begin{equation}\label{IP}
\mathcal{S}(x|\theta)=\log \omega(x|\theta)\equiv-\mathfrak{I}(x|\theta),
\end{equation}
the metric tensor (\ref{cov.equation}) can be rewritten as follows \cite{Vel.FG}:
\begin{equation}\label{metric}
g_{ij}(x|\theta)=-D_{i}D_{j}\mathcal{S}(x|\theta)=-\frac{\partial^{2}\mathcal{S}(x|\theta)}{\partial x^{i}\partial x^{j}}+\Gamma^{k}_{ij}(x|\theta)\frac{\partial\mathcal{S}(x|\theta)}{\partial x^{k}}.
\end{equation}
Here, $D_{i}$ is the \emph{covariant derivative} associated with the Levi-Civita affine connections $\Gamma^{k}_{ij}(x|\theta)$:
\begin{equation}
\Gamma _{ij}^{k}\left( x|\theta \right) =g^{km}(x|\theta)\frac{1}{2}\left[\frac{%
\partial g_{im}(x|\theta)}{\partial x^{j}}+\frac{\partial g_{jm}(x|\theta)}{\partial x^{i}}-%
\frac{\partial g_{ij}(x|\theta)}{\partial x^{m}}\right].  \label{Levi-Civita}
\end{equation}
The alternative form (\ref{metric}) of problem (\ref{cov.equation}) clearly evidences the covariant character of this set of partial differential equations. According to expression (\ref{metric}), the metric tensor $g_{ij}(x|\theta)$ defines a positive definite distance notion (\ref{fluct.dist}), while the information potential $\mathcal{S}(x|\theta)$ is locally concave everywhere. This last behavior guarantees the uniqueness of the point $\bar{x}$ where the information potential reaches a global maximum, that is, the uniqueness of the point of global maximum $\bar{x}$ of the probabilistic weight $\omega(x|\theta)$.

The main consequence derived from equation (\ref{metric}) is the possibility to rewrite the distributions family (\ref{DF2}) into the following \textit{Riemannian gaussian representation} \cite{Vel.GEO,Vel.FG}:
\begin{equation}\label{UGR}
dp(x|\theta)=\frac{1}{\mathcal{Z}(\theta)}\exp\left[-\frac{1}{2}\ell^{2}_{\theta}(x,\bar{x})\right]d\mu(x|\theta),
\end{equation}
where $\ell_{\theta}(x,\bar{x})$ denotes the \emph{separation distance} between the arbitrary point $x$ and the point $\bar{x}$ with maximum information potential $\mathcal{S}(x|\theta)$ (the arc-length $\Delta s$ of the \emph{geodesics} that connects these points). Moreover, the negative of the logarithm of gaussian partition function $\mathcal{Z}(\theta)$ defines the so-called \textit{gaussian potential}:
\begin{equation}\label{gaussianpotential}
\mathcal{P}(\theta)=-\log\mathcal{Z}(\theta),
\end{equation}
which appears as the first integral of the problem (\ref{metric}):
\begin{equation}\label{Sdecomposition}
\mathcal{P}(\theta )=\mathcal{S}(x|\theta )+\frac{1}{2}\psi^{2}(x|\theta ).
\end{equation}
Here, $\psi^{2}(x|\theta)=\psi^{i}\left( x|\theta \right)\psi _{i}\left( x|\theta \right)=g^{ij}(x|\theta)\psi_{i}\left( x|\theta \right)\psi _{j}\left( x|\theta \right)$ is the square norm of covariant vector field $\psi _{i}\left( x|\theta \right)$ defined by the gradient of the information potential $\mathcal{S}\left( x|\theta\right)$:
\begin{equation}
\psi _{i}\left( x|\theta \right) =-D_{i}\mathcal{S}\left( x|\theta
\right) \equiv-\partial \mathcal{S}\left( x|\theta \right) /\partial
x^{i}.  \label{cov.DGF}
\end{equation}
The factor $2\pi$ of definition (\ref{inv.volume}) guarantees that the gaussian partition function $\mathcal{Z}(\theta)$ drops the unity when the Riemannian structure of the manifold $\mathcal{M}$ is the same of Euclidean real space $\mathbb{R}^{n}$.

Riemannian gaussian representation (\ref{UGR}) can be obtained combining equations (\ref{DF2}) and (\ref{Sdecomposition}) with the following the identity:
\begin{equation}\label{ident.HJeq}
\psi^{2}(x|\theta)\equiv \ell^{2}_{\theta}(x,\bar{x}).
\end{equation}
This last relation is a consequence of the \emph{geodesic character} of the curves $x_{g}(s)\in\mathcal{M}$ derived from the following set of ordinary differential equations \cite{Vel.FG}:
\begin{equation}\label{hydro.equations}
\frac{dx^{i}_{g}(s)}{ds}=\upsilon^{i}\left[x_{g}(s)|\theta\right].
\end{equation}
Here, $\upsilon^{i}(x|\theta)=g^{ij}(x|\theta)\upsilon_{j}(x|\theta)$ is the contravariant form of the unitary vector field $\upsilon_{i}(x|\theta)$ associated with the vector field (\ref{cov.DGF}):
\begin{equation}
\upsilon_{i}(x|\theta)=\psi_{i}\left( x|\theta \right)/\psi\left(x|\theta \right),
\end{equation}
while the parameter $s$ is the arc-length of the curve $x_{g}(s)$. It is easy to check that this unitary vector field obeys the geodesic differential equation:
\begin{equation}
\upsilon^{j}(x|\theta)D_{j}\upsilon_{i}(x|\theta)=\upsilon^{j}(x|\theta)\left[g_{ij}(x|\theta)-\upsilon_{i}(x|\theta)\upsilon_{j}(x|\theta)\right]\equiv0.
\end{equation}
Identity (\ref{ident.HJeq}) follows from the directional derivatives:
\begin{equation}
\frac{d\mathcal{S}\left( x_{g}(s)|\theta \right)}{ds}\equiv\psi(x_{g}(s)|\theta)\mbox{ and }\frac{d^{2}\mathcal{S}\left( x_{g}(s)|\theta \right)}{ds^{2}}\equiv-1,
\end{equation}
which can be obtained from equation (\ref{hydro.equations}).

Riemannian gaussian representation (\ref{UGR}) rephrases the distributions family (\ref{DF}) in term of geometric notions of the manifold $\mathcal{M}$. According to this result, the distance $\ell_{\theta}(x,\bar{x})$ is a measure of the \emph{occurrence probability} of a deviation from the state $\bar{x}$ with maximum information potential. At first glance, \emph{gaussian distributions} exhibit a very special status within fluctuation geometry, overall, because of any continuous distribution function can be rephrased as a generalized gaussian distribution defined on a Riemannian manifold. As shown in the next section, equation (\ref{UGR}) is a key result to understand the statistical relevance of the curvature tensor of fluctuation geometry.

\section{Curvature of the statistical manifold $\mathcal{M}$}\label{CurvatureSect}

\subsection{Curvature tensor of Riemannian geometry}

The affine connections $\Gamma _{ij}^{k}=\Gamma _{ij}^{k}(x|\theta)$ are employed to introduce of the \textit{curvature tensor} $ R^{l}_{ijk}=R^{l}_{ijk}(x|\theta)$ of the manifold $\mathcal{M}$:
\begin{equation}\label{curvature}
R^{l}_{ijk}=\frac{\partial}{\partial X^{i}}\Gamma^{l}_{jk}-\frac{\partial}{%
\partial X^{j}}\Gamma^{l}_{ik}+\Gamma^{l}_{im}\Gamma^{m}_{jk}-
\Gamma^{l}_{jm}\Gamma^{m}_{ik}.
\end{equation}
In general, the affine connections $\Gamma _{ij}^{k}(x|\theta)$ and the metric tensor $g_{ij}\left( x|\theta \right)$ are independent entities of Riemannian geometry. However, the knowledge of the metric tensor allows to introduce natural affine connections: the Levi-Civita connections (\ref{Levi-Civita}). These affine connections are also referred to in the literature as the metric connections or the Christoffel symbols. The same ones follow from the consideration of a \emph{torsion-free covariant differentiation} $D_{i}$ that obeys the \emph{condition of Levi-Civita parallelism} \cite{Berger}:
\begin{equation}
D_{k}g_{ij}\left( x|\theta \right) =0.  \label{Dg}
\end{equation}
Using the Levi-Civita connections, the curvature tensor can be expressed in terms of the metric tensor $g_{ij}(x|\theta)$ and its first and second partial derivatives. For example, its fourth-rank covariant form $R_{ijkl}=g_{lm}R^{m}_{ijk}$ adopts the following form:
\begin{eqnarray}\label{curvature.2}
R_{ijkl}=\frac{1}{2}\left(\frac{\partial^{2}g_{il}}{\partial
x^{j}\partial x^{k}}+\frac{\partial^{2}g_{jk}}{\partial
x^{i}\partial x^{l}}-\frac{\partial^{2}g_{jl}}{\partial
x^{i}\partial x^{k}}-\frac{\partial^{2}g_{ik}}{\partial
x^{j}\partial x^{l}}\right)+\\+g_{mn}\left(\Gamma^{m}_{il}\Gamma^{n}_{jk}-\Gamma^{m}_{jl}\Gamma^{n}_{ik}\right).\nonumber
\end{eqnarray}
Additionally, one can introduce the \textit{Ricci curvature tensor} $R_{ij}(x|\theta)$:
\begin{equation}\label{Ricci}
R_{ij}(x|\theta)=R_{kij}^{k}(x|\theta)
\end{equation}
as well as the \textit{curvature scalar} $R(x|\theta)$:
\begin{equation} \label{scalar.curvature}
R(x|\theta)=g^{ij}(x|\theta)R^{k}_{kij}(x|\theta)=g^{ij}(x|\theta)g^{kl}(x|\theta)R_{kijl}(x|\theta).
\end{equation}
According to Riemannian geometry \cite{Berger}, the curvature scalar $R(x|\theta)$ is the only invariant derived from the first and second partial derivatives of the metric tensor $g_{ij}(x|\theta)$.

The curvature tensor characterizes the deviation of local geometric properties of a manifold $\mathcal{M}$ from the properties of the Euclidean geometry. For example, the volume of a small sphere about a point $x$ has smaller (larger) volume (area) than a sphere of the same radius defined on an Euclidean manifold $\mathbb{E}^{n}$ when the scalar curvature $R(x|\theta)$ is positive (negative) at that point. Quantitatively, this behavior is described by the following approximation formulae:
\begin{eqnarray}\label{AF1}
  \frac{\mathrm{Vol} \left[\mathbb{S}^{(n-1)}(x|\ell)\subset\mathcal{M}\right] }{\mathrm{Vol} \left[\mathbb{S}^{(n-1)}(x|\ell)\subset\mathbb{E}^{n}\right]} &=& 1-\frac{R(x|\theta)}{6(n+2)}\ell^{2} +O(\ell^{4}),\\
  \frac{\mathrm{Area} \left[\mathbb{S}^{(n-1)}(x|\ell)\subset\mathcal{M}\right]}{\mathrm{Area} \left[\mathbb{S}^{(n-1)}(x|\ell)\subset\mathbb{E}^{n}\right]} &=& 1-\frac{R(x|\theta)}{6n}\ell^{2}+O(\ell^{4}),\label{AF2}
\end{eqnarray}
where the notation $\mathbb{S}^{(m)}(x|\ell)$ represents a $m$-dimensional sphere with small radius $\ell$ centered at the point $x$. Accordingly, the local effects associated with the curvature of the manifold $\mathcal{M}$ appears as second-order (and higher) corrections of the Euclidean formulae. The best known example of Euclidean manifold is the $n$-dimensional Euclidean real space $\mathbb{R}^{n}$. The geometry defined on surface of cylinder $\mathbb{C}^{(2)}\in\mathbb{R}^{3}$ is other example of Euclidean geometry, while the geometry defined on the surface of the $n$-dimensional sphere $\mathbb{S}^{(n)}\in\mathbb{R}^{n+1}$ with $n \geq 2$ is a typical example of curved geometry (with a constant positive curvature).

Some tensorial identities can be easily demonstrated by adopting the so-called \textit{normal coordinates}. For the sake of simplicity, let us assume that the point of interest of the manifold $\mathcal{M}$ corresponds to the origin $x=(0,0,\ldots 0)$. Moreover, let us also assume that the metric tensor components and their first partial derivatives in that point satisfy the following conditions:
\begin{equation}
g_{ij}(0|\theta)=\delta_{ij}\mbox{ and }\partial g_{ij}(0|\theta)/\partial x^{k}=0,
\end{equation}
with $\delta_{ij}$ being the Kronecker delta. The coordinate representation defined by the previous conditions represent a normal coordinates centered at the origin. Since the Levi-Civita connections vanishing at that point, $\Gamma^{k}_{ij}(0|\theta)=0$, the calculation of the curvature tensor $R_{ijkl}(0|\theta)$ only involves the second derivatives of the metric tensor:
\begin{equation}\label{second.curvature}
R_{ijkl}(0|\theta)=\frac{1}{2}\left(\frac{\partial^{2}g_{il}(0|\theta)}{\partial
x^{j}\partial x^{k}}+\frac{\partial^{2}g_{jk}(0|\theta)}{\partial
x^{i}\partial x^{l}}-\frac{\partial^{2}g_{jl}(0|\theta)}{\partial
x^{i}\partial x^{k}}-\frac{\partial^{2}g_{ik}(0|\theta)}{\partial
x^{j}\partial x^{l}}\right).
\end{equation}
Using normal coordinates, the distance metric and the first covariant derivatives at the origin behaves as their Euclidean counterparts. A remarkable result (due to Riemann himself) associated with normal coordinates is the following second-order approximation for the distance notion \cite{Berger}:
\begin{equation}\label{develop}
g_{ij}(x|\theta)dx^{i}dx^{j}=dx^{i}dx^{i}+\frac{1}{12}R_{imjn}(0|\theta)dS^{im}dS^{jn}+O\left(\left|x\right|^{2}\right),
\end{equation}
where $dS^{ij}=x^{j}dx^{i}-x^{i}dx^{j}$. Accordingly, a curved Riemannian manifold locally looks-like an Euclidean manifold at zeroth and first-order approximation of the power-expansion using normal coordinates, while the local curvature of this manifold appears as a \emph{second-order effect}. Normal coordinates will be employed to develop the second-order geometric expansion of a distribution function, which is a statistic counterpart of asymptotic geometric formulae (\ref{AF1}) and (\ref{AF2}).

\subsection{Curvature tensor and the irreducible statistical correlations}
Previously, it was shown that distributions families whose manifolds $\mathcal{M}$ are diffeomorphic to the one-dimensional real space $\mathbb{R}$ are diffeomorphic distributions. However, this property cannot be extended to distributions families whose statistical manifolds $\mathcal{M}$ have a dimension $n\geq 2$.

\begin{remark}
Two distributions families $dp_{1}(x_{1}|\theta)$ and $dp_{2}(x_{2}|\theta)$ whose abstract statistical manifolds $\mathcal{M}_{1}$ and $\mathcal{M}_{2}$ have a dimension $n\geq 2$ are not necessarily diffeomorphic distributions.
\end{remark}
\begin{proof}
A \emph{diffeomorphism} is a map that \emph{preserves both the differential and Riemannian structures}. Thus, if two distributions families $dp_{1}(x_{1}|\theta)$ and $dp_{2}(x_{2}|\theta)$ have statistical manifolds $\mathcal{M}_{1}$ and $\mathcal{M}_{2}$ with different Riemannian structures, their respective complete sets of random quantities $\xi_{1}$ and $\xi_{2}$ are not related by a diffeomorphism, $\xi_{2}\neq\phi(\xi_{1})$. Precisely, two statistical manifolds $\mathcal{M}_{1}$ and $\mathcal{M}_{2}$ with dimension $n\geq 2$ can differ in regard to their \emph{curvatures}.
\end{proof}

Curvature notion plays a relevant role in fluctuation geometry. Besides the question about whether or not two continuous distributions are diffeomorphic distributions, curvature tensor appears as indicator about the existence or nonexistence of irreducible statistical correlations.

\begin{definition}\label{Defin.StatInd}
A continuous distribution $dp(x|\theta)$ exhibits a \textbf{reducible statistical dependence} if it possesses a diffeomorphic distribution $dp(\check{x}|\theta)$ that admits to be decomposed into independent distribution functions $dp^{(i)}(\check{x}^{i}|\theta)$ for each coordinate as follows:
\begin{equation}\label{factorization}
dp(\check{x}|\theta)=\prod^{n}_{i=1}dp^{(i)}(\check{x}^{i}|\theta).
\end{equation}
Otherwise, the distribution function $dp(x|\theta)$ exhibits an \textbf{irreducible statistical dependence}.
\end{definition}
\begin{example}\label{example.XY}
The following continuous distribution:
\begin{equation}\label{example.ind}
dp(x,y)=A\exp\left[-x^{2}-y^{2}-xy\right]dxdy
\end{equation}
describes a statistical dependence between the coordinates $x$ and $y$. However, this distribution can be rewritten into independent distributions:
\begin{equation}
dp(\check{x},\check{y})=\sqrt{\frac{3}{2\pi}}\exp\left[-\frac{3}{2}
\check{x}^{2}\right]d\check{x}\frac{1}{\sqrt{2\pi}}\exp\left[-\frac{1}{2}\check{y}^{2}\right]
d\check{y}
\end{equation}
considering the coordinate change $(\check{x},\check{y})=\phi(x,y)$:
\begin{equation}
\check{x}=\frac{1}{\sqrt{2}}(x+y),\:\check{y}=\frac{1}{\sqrt{2}}(x-y).
\end{equation}
Therefore, distribution (\ref{example.ind}) exhibits a reducible statistical dependence.
\end{example}

\begin{example}\label{GDF}
Distribution (\ref{example.ind}) is a particular case of the gaussian family:
\begin{equation}\label{GaussianFamily}
dp_{G}(x|\theta)=\exp\left[-\frac{1}{2}\sigma_{ij}(x^{i}-\mu^{i})(x^{j}-\mu^{j})
\right]\sqrt{\left|\frac{\sigma_{ij}}{2\pi}\right|}dx,
\end{equation}
where the control parameters $\theta=(\mu^{i},\sigma_{ij})$ are the means $\mu^{i}=\left\langle x^{i}\right\rangle$ and the inverse matrix $\sigma_{ij}$ of the self-correlations $\sigma^{ij}=\left\langle\delta x^{i}\delta x^{j}\right\rangle$. It is easy to realize that the metric tensor for this distributions family is $g_{ij}(x|\theta)\equiv \sigma_{ij}=const$. The coordinates $x=(x^{1}, x^{2}, \ldots x^{n})$ can be subjected to a translation-rotation coordinate change $x^{i}=\mu^{i}+T^{i}_{j}\check{x}^{i}$ that ensures the diagonal character of the new self-correlation matrix $\tilde{\sigma}^{ij}=\left\langle\delta \check{x}^{i}\delta \check{x}^{j}\right\rangle=(\sigma^{i})^{2}\delta^{ij}$. Thus, distribution function resulting from this coordinate change can be decomposed into independent distributions:
\begin{equation}
dp(\check{x}|\theta)=\prod^{n}_{i=1}\exp\left[-\frac{1}{2}\left(\frac{\check{x}^{i}}
{\sigma^{i}}\right)^{2}\right]\frac{d\check{x}^{i}}{\sigma^{i}\sqrt{2\pi}}.
\end{equation}
Gaussian family (\ref{GaussianFamily}) exhibits a reducible statistical dependence. By definition, any diffeomorphic distribution of the gaussian family (\ref{GaussianFamily}) will also exhibit a reducible statistical dependence.
\end{example}

As already evidenced in the previous examples, the \textit{self-correlation matrix} $\sigma^{ij}$:
\begin{equation}\label{correlation}
\sigma^{ij}=\mathrm{cov}(x^{i},x^{j})=\left\langle(x^{i}-\left\langle x^{i}\right\rangle)(x^{j}-\left\langle
x^{j}\right\rangle)\right\rangle
\end{equation}
among the coordinates $x=(x^{1}, x^{2}, \ldots x^{n})$ of a given coordinate representation $\mathcal{R}_{x}$ of the manifold $\mathcal{M}$ cannot be employed to indicate the existence of irreducible statistical dependence. For an arbitrary family of distributions (\ref{DFxi}), the self-correlation matrix (\ref{correlation}) does not represent a tensorial quantity of any kind. Therefore, these quantities are unsuitable to predict the existence (or nonexistence) of a reducible statistical dependence. On the contrary, the statistical manifold $\mathcal{M}$ associated with the gaussian family (\ref{GaussianFamily}) exhibits a vanishing curvature tensor $R^{l}_{ijk}(x|\theta)=0$, a property that is protected by the covariant transformation rules of the curvature tensor:
\begin{equation}
\check{R}^{p}_{mno}(\check{x}|\theta)=R^{l}_{ijk}(x|\theta)\frac{\partial x^{i}}{\partial \check{x}^{m}}\frac{\partial x^{j}}{\partial \check{x}^{n}}\frac{\partial x^{k}}{\partial \check{x}^{o}}\frac{\partial \check{x}^{p}}{\partial x^{l}}.
\end{equation}
The statistical manifold $\mathcal{M}$ associated with the gaussian family (\ref{GaussianFamily}) is flat, that is, it exhibits the same Riemannian structure of Euclidean $n$-dimensional real space $\mathbb{R}^{n}$. This example strongly suggests the existence of a direct connection between the existence of reducible statistical dependence and the curvature tensor $R^{l}_{ijk}(x|\theta)$ of the statistical manifold $\mathcal{M}$. Remarkably, such a connection is almost a trivial question from the viewpoint of  Riemannian geometry.

\begin{proposition}\label{Prop.Descomp}
The existence (or nonexistence) of a reducible statistical dependence for a given distributions family (\ref{DFxi}) is reduced to the existence (or nonexistence) of a \textbf{Cartesian decomposition} of its associated statistical manifold $\mathcal{M}$ into two (or more) independent statistical manifolds $\left\{\mathcal{A}^{(i)}_{\theta}\right\}$:
\begin{equation}\label{decompositionM}
\mathcal{M}=\mathcal{A}^{(1)}\otimes\mathcal{A}^{(2)}\ldots\otimes\mathcal{A}^{(l)}.
\end{equation}
\end{proposition}

\begin{proof}
\textbf{Cartesian product of Riemannian manifolds} is a generalization of Cartesian product of spaces that includes the differential and the Riemannian structures. In particular, the distance notion (\ref{fluct.dist}) of the statistical manifold  $\mathcal{M}$ is determined from the distance notions $ds^{2}_{(k)}=g^{(k)}_{i_{k}j_{k}}(a_{k}|\theta)da_{k}^{i_{k}}da_{k}^{j_{k}}$ of each manifold $\mathcal{A}^{(k)}$ via \emph{Pythagorean theorem} as follows:
\begin{equation}
ds^{2}=ds^{2}_{(1)}\bigoplus ds^{2}_{(2)}\ldots\bigoplus ds^{2}_{(l)}\equiv \sum^{l}_{k=1}ds^{2}_{(k)}.
\end{equation}
Let us denote by $\mathcal{R}_{a_{k}}$ a certain coordinate representation of the manifold $\mathcal{A}^{(k)}$. Given a statistical manifold $\mathcal{M}$ and its Riemannian structure, the essential property allowing Cartesian decomposition as (\ref{decompositionM}) is that the metric tensor $g_{ij}(x|\theta)$ exhibits the following matrix form:
\begin{equation}\label{g.matrix}
g_{ij}(x|\theta)=\left(
                   \begin{array}{ccccc}
                     g^{(1)}_{i_{1}j_{1}}(a_{1}|\theta) & 0 & \ldots & \ldots & 0 \\
                     0 & g^{(2)}_{i_{2}j_{2}}(a_{2}|\theta) & 0 & \ldots & 0 \\
                     \vdots & 0 & \ddots &  &\vdots \\
                     \vdots & \vdots &  & g^{(l-1)}_{i_{l-1}j_{l-1}}(a_{l-1}|\theta) & 0 \\
                     0 & 0 & \ldots & 0 & g^{(l)}_{i_{l}j_{l}}(a_{l}|\theta) \\
                   \end{array}
                 \right)
\end{equation}
for a certain coordinate representation $\mathcal{R}_{x}=\mathcal{R}_{a_{1}}\otimes \mathcal{R}_{a_{2}}\ldots\otimes \mathcal{R}_{a_{l}}$ of the manifold $\mathcal{M}$. As expected, the underlying Cartesian decomposition imposes some \emph{composition rules} for tensorial quantities defined on the manifold $\mathcal{M}$ in term of corresponding tensorial entities for each statistical manifold $\mathcal{A}^{(k)}$. For example, equation (\ref{g.matrix}) implies the additive character of the statistical distance $\ell^{2}_{\theta}(x,\bar{x})$ and the factorization of the invariant volume element $d\mu(x|\theta)$ as follows:
\begin{equation}\label{desc.elle}
\ell^{2}_{\theta}(x,\bar{x})=\sum^{l}_{k=1}\ell^{2}_{\theta}(a_{k},\bar{a}_{k})\mbox{ and }d\mu(x|\theta)=\prod^{l}_{k=1}d\mu^{(k)}(a_{k}|\theta),
\end{equation}
where $(\bar{a}_{1},\ldots,\bar{a}_{l})$ are the coordinates of the point $\bar{x}$ with maximum information potential, and $d\mu^{(k)}(a_{k}|\theta)$ the invariant volume element of the manifold $\mathcal{A}^{(k)}$:
\begin{equation}
d\mu^{(k)}(a_{k}|\theta)=\sqrt{\left|g^{(k)}_{i_{k}j_{k}}(a_{k}|\theta)/2\pi\right|}da_{k}.
\end{equation}
Using the Riemannian gaussian representation (\ref{UGR}) and the relations (\ref{desc.elle}), one immediately obtains the composition rule of the probability distribution (\ref{DF2}) into independent distributions:
\begin{equation}\label{factor.dp}
dp(x|\theta)=\prod^{l}_{k=1}dp^{(k)}(a_{k}|\theta),
\end{equation}
where $dp^{(k)}(a_{k}|\theta)$ is the probability distribution:
\begin{equation}
dp^{(k)}(a_{k}|\theta)=\frac{1}{\mathcal{Z}^{(k)}(\theta)}\exp\left[-\frac{1}{2}\ell^{2}_{\theta}(a_{k},\bar{a}_{k})\right]d\mu^{(k)}(a_{k}|\theta).
\end{equation}
The composition rule (\ref{factor.dp}) imposes the factorization of the gaussian partition function $\mathcal{Z}(\theta)$:
\begin{equation}
\mathcal{Z}(\theta)=\prod^{l}_{k=1}\mathcal{Z}^{(k)}(\theta),
\end{equation}
and hence, the additive character of the gaussian potential $\mathcal{P}(\theta)$ and the information potential $\mathcal{S}(x|\theta)$:
\begin{equation}
\mathcal{P}(\theta)=\sum^{l}_{k=1}\mathcal{P}^{(k)}(\theta)\Rightarrow\mathcal{S}(x|\theta)=\sum^{l}_{k=1}\mathcal{S}^{(k)}(a_{k}|\theta),
\end{equation}
where $\mathcal{P}^{(k)}(\theta)$ and $\mathcal{S}^{(k)}(a_{k}|\theta)$ are given by:
\begin{equation}
\mathcal{P}^{(k)}(\theta)=-\log \mathcal{Z}^{(k)}(\theta)\mbox{ and }\mathcal{S}^{(k)}(a_{k}|\theta)=\mathcal{P}^{(k)}(\theta)-\frac{1}{2}\ell^{2}_{\theta}(a_{k},\bar{a}_{k}).
\end{equation}
Thus, the existence of a Cartesian decomposition (\ref{decompositionM}) for the statistical manifold $\mathcal{M}$ implies the decomposition of the distributions family (\ref{DF}) into a set of independent distribution functions.
\end{proof}

\begin{definition}
A given manifold $\mathcal{A}$ is said to be an \textbf{irreducible manifold} when the same one does not admit the Cartesian decomposition (\ref{decompositionM}). Moreover, a given Cartesian decomposition (\ref{decompositionM}) is said to be an \textbf{irreducible Cartesian decomposition} if each independent manifold  $\mathcal{A}^{(k)}$ is an irreducible manifold.
\end{definition}

\begin{theorem}\label{Main}
The flat character of the statistical manifold $\mathcal{M}$ implies the existence of a reducible statistical dependence for the family of distributions (\ref{DFxi}), while its curved character implies the existence of an irreducible statistical dependence.
\end{theorem}

\begin{proof}
The existence of a reducible statistical dependence (\ref{factorization}) is a very strong restriction. This property demands that the associated manifold $\mathcal{M}$ admits an irreducible Cartesian decomposition (\ref{decompositionM}) into a set of one-dimensional manifolds $\left\{\mathcal{A}^{(k)}\right\}$, $k=(1,2,\ldots n)$. The matrix representation of the metric tensor (\ref{g.matrix}) imposes the vanishing of those the components of curvature tensor $R_{ijkl}(x|\theta)$ involving indexes belonging to different manifolds in Cartesian decomposition (\ref{decompositionM}). Consequently, the existence of a reducible statistical dependence implies the vanishing of all components of the curvature tensor  $R_{ijkl}(x|\theta)$. Alternatively, if the statistical manifold $\mathcal{M}$ exhibits a vanishing curvature tensor, its Riemannian structure is the same of an Euclidean $n$-dimensional manifold $\mathbb{E}^{n}$. Since any Euclidean manifold $\mathbb{E}^{n}$ admits an irreducible Cartesian decomposition into a set of one-dimensional manifolds, the existence of a vanishing curvature tensor also implies the existence of a reducible statistical dependence for the distributions family (\ref{DFxi}) (in accordance with \textbf{Proposition \ref{Prop.Descomp}}). Let us now consider the case where the manifold $\mathcal{M}$ exhibits a non-vanishing curvature tensor. Since Cartesian product $\mathcal{A}\otimes\mathcal{B}$ of two arbitrary one-dimensional manifolds $\mathcal{A}$ and $\mathcal{B}$ has always a vanishing curvature tensor, any curved two-dimensional manifold is irreducible. Consequently, if $\mathcal{M}$ is a curved two-dimensional manifold, its associated distributions family (\ref{DFxi}) exhibits an irreducible statistical dependence. If the manifold $\mathcal{M}$ has a dimension $n\geq 2$, the irreducible statistical dependence of the distributions family (\ref{DF}) implies that the irreducible Cartesian decomposition (\ref{decompositionM}) must contain, at least, an irreducible statistical manifold $\mathcal{A}^{(k)}$ with dimension $d=\mathrm{dim}\left[\mathcal{A}^{(k)}\right]\geq 2$ with non-vanishing curvature. In general, the question about the Cartesian decomposition of a Riemannian manifold into independent manifolds with arbitrary dimensions is better phrased and understood in the language of \emph{holonomy groups}. The relation of holonomy of a connection with the curvature tensor is the main content of \emph{Ambrose-Singer theorem}, while \emph{de Rham theorem} states the conditions for a global Cartesian decomposition \cite{Berger}.
\end{proof}

\subsection{Second-order geometric expansion}
Gaussian family (\ref{GaussianFamily}) plays a relevant role in statistical and physical applications. In particular, this family of distributions (\ref{GaussianFamily}) represents \emph{asymptotic distributions} in some appropriate limits, such as the case of \emph{central limit theorem} in statistics \cite{Dekking} or the fluctuating behavior of large thermodynamic systems in Einstein's fluctuation theory \cite{Reichl}. The statistical manifold $\mathcal{M}$ associated with the gaussian family (\ref{GaussianFamily}) exhibits the same Riemannian structure of Euclidean $n$-dimensional real space $\mathbb{R}^{n}$. The asymptotic convergence of an arbitrary distributions family (\ref{DFxi}) towards a gaussian family is a consequence of the weakening of curvature at a small neighborhood of the point $\bar{x}$ with maximum information potential $\mathcal{S}(x|\theta)$. In general, the geometric properties of a small region of a curved manifold $\mathcal{M}$ are approximately Euclidean if the linear dimension $\ell$ of this region is sufficiently small. This asymptotic behavior is expressed in the approximation formulae (\ref{AF1}) and (\ref{AF2}). Gaussian family (\ref{GaussianFamily}) always arises as the \emph{Euclidean} or \emph{zeroth-order approximation} of any distributions family (\ref{DF}). The effects of the curved character of statistical manifold $\mathcal{M}$ are manifested as \emph{second-order corrections} of the gaussian approximation. The study of such a geometric power-expansion is the main goal of the present subsection. In inference theory, the counterpart approach of this geometric expansion is referred to as \emph{higher-order asymptotic theory of statistical estimation} \cite{Amari}.

\begin{lemma}\label{lemma2}
Riemannian gaussian representation (\ref{UGR}) can be expressed into the following \textbf{spherical coordinate representation}:
\begin{equation}\label{polar}
dp(\ell,q|\theta)=\frac{1}{\mathcal{Z}_{g}(\theta)}\frac{1}{\sqrt{2\pi}}\exp\left(-\frac{1}{2}\ell^{2}\right)
d\ell d\Sigma_{g}(q|\ell,\theta).
\end{equation}
Here, $d\Sigma(q|\ell,\theta)$ is the hyper-surface element:
\begin{equation}
d\Sigma_{g}(q|\ell,\theta)=\sqrt{\left|g_{\alpha\beta}(\ell,q|\theta)/2\pi\right|}dq.
\end{equation}
obtained from the metric tensor $g_{\alpha\beta}(\ell,q|\theta)$ associated with the projected Riemannian structure on the surface of the $(n-1)$-dimensional sphere $\mathbb{S}^{(n-1)}(\bar{x}|\ell)\subset\mathcal{M}$ with radius $\ell$.
\end{lemma}
\begin{proof}
Let us consider the geodesic family $x_{g}(s;\mathbf{e})$ derived from the set of ordinary differential equations (\ref{hydro.equations}). The quantities $\mathbf{e}=\left\{e^{i}\right\}$ represent the asymptotic values of the unitary vector field $\upsilon^{i}(x|\theta)$ at the point $\bar{X}$ with maximum information potential:
\begin{equation}
e^{i}=\lim_{s\rightarrow 0}\frac{dx^{i}_{g}(s;\mathbf{e})}{ds}.
\end{equation}
Here, the parameter $s\equiv \ell$ is the arc-length of these geodesics with reference to the origin point $\bar{x}$. By definition, the vector field $\upsilon^{i}(x|\theta)$ is a normal unitary vector of the surface of constant information potential $\mathcal{S}(x|\theta)$. Moreover, such a surface is just the $(n-1)$-dimensional sphere $\mathbb{S}^{(n-1)}(\bar{x}|\ell)\subset\mathcal{M}$ with radius $\ell$ and centered at the point $\bar{x}$. The vectors $\mathbf{e}=\left\{e^{i}\right\}$  can be parameterized as $\mathbf{e}=\mathbf{e}(q)$ using the intersection point $q$ of the geodesics $x_{g}(s;\mathbf{e})$ with the sphere $\mathbb{S}^{(n-1)}(\bar{x}|\ell)\subset\mathcal{M}$. One can employ the variables $\rho=(\ell,q)$ to introduce a spherical coordinate representation $\mathcal{R}_{\rho}$ centered at the point $\bar{x}$ with maximum information potential $\mathcal{S}(x|\theta)$. The coordinate change $\phi:\mathcal{R}_{x}\rightarrow\mathcal{R}_{\rho}$ is defined from the geodesic family $x=x^{i}_{g}[\ell|\mathbf{e}(q)]$, whose partial derivatives are given by:
\begin{equation}
\upsilon^{i}(x|\theta)=\frac{\partial
x^{i}_{g}[\ell|\mathbf{e}(q)]}{\partial
\ell},\tau^{i}_{\alpha}(x|\theta)=\frac{\partial
x^{i}_{g}[\ell|\mathbf{e}(q)]}{\partial q^{\alpha}}.
\end{equation}
The new $(n-1)$ vector fields $\tau^{i}_{\alpha}(x|\theta)$ are perpendicular to the unitary vector field $\upsilon^{i}(x|\theta)$
because of they are tangential vectors of the sphere $\mathbb{S}^{(n-1)}(\bar{x}|\ell)$. Consequently, the non-vanishing components of the metric tensor written in this spherical coordinate representation are given by:
\begin{eqnarray}\label{g.new}
g_{\ell\ell}(\ell,q|\theta)=1, g_{\alpha\beta}(\ell,q|\theta)=g_{ij}(x|\theta)\tau^{i}_{\alpha}(x|\theta)\tau^{j}_{\beta}(x|\theta).
\end{eqnarray}
Here, $g_{\alpha\beta}(\ell,q|\theta)$ is the metric tensor that defines the projected Riemannian structure on the sphere $\mathbb{S}^{(n-1)}(\bar{x}|\ell)$. Equation (\ref{polar}) is straightforwardly obtained using relations (\ref{g.new}). Any coordinate change considered in the framework of the spherical coordinate representation (\ref{polar}) only involves the spherical variables $q$ because of the radial variable $\ell$ is invariant quantity. As expected, the spherical coordinate representation (\ref{polar}) is singular at the point $\ell=0$, that is, the all points $(\ell,q)$ with $\ell=0$ corresponds to the point $\bar{x}$ without mattering about the values of the spherical coordinates $q$.
\end{proof}
\begin{theorem}
Spherical coordinate representation (\ref{polar}) obeys the following asymptotic distribution for $\ell$ sufficiently small:
\begin{equation}\label{asymptotic}
dp(\ell,q|\theta)=\frac{1}{\mathcal{Z}_{g}(\theta)}\left[1-
\frac{1}{24}\ell^{2}\mathcal{F}(q|\theta)+O(\ell^{4})\right]dp_{G}(\ell,q|\theta).
\end{equation}
Here, $dp_{G}(\ell,q|\theta)$ denotes the spherical coordinate representation of a gaussian distribution associated with the local Euclidean properties of the manifold $\mathcal{M}$ at the point $\bar{x}$:
\begin{equation}\label{SCRGaussianFamily}
dp_{G}(\ell,q|\theta)=\exp\left(-\frac{1}{2}\ell^{2}\right) \frac{\ell^{n-1}d\ell}{\sqrt{2\pi}}\sqrt{\left|\frac{\kappa_{\alpha\beta}(q)}{2\pi}\right|}dq.
\end{equation}
where $\kappa_{\alpha\beta}(q)=\bar{g}_{ij}\xi^{i}_{\alpha}(q)\xi^{j}_{\beta}(q)$. The $(n-1)$ vector fields $\mathbf{\xi}_{\alpha}(q)=\left\{\xi^{i}_{\alpha}(q)\right\}$ are obtained from the unitary vector field $\mathbf{e}(q)=\left\{e^{i}(q)\right\}$ of \textbf{Lemma \ref{lemma2}} at the point $\bar{x}$ as follows:
\begin{equation}\label{xi}
\xi^{i}_{\alpha}(q)=\frac{\partial e^{i}(q)}{\partial q^{\alpha}}.
\end{equation}
$\mathcal{F}(q|\theta)$ is a function on the spherical coordinates $q$ defined as follows:
\begin{equation}\label{spherical}
\mathcal{F}(q|\theta)=\bar{R}_{ijkl}\kappa^{\alpha\beta}(q)S^{ij}_{\alpha}(q)S^{kl}_{\beta}(q),
\end{equation}
which is hereinafter referred to as the \textbf{spherical function}. Moreover, $\bar{R}_{ijkl}=R_{ijkl}(\bar{x}|\theta)$ is the curvature tensor (\ref{curvature.2}) evaluated at the point $\bar{x}$, while the quantities $S^{ij}_{\alpha}(q)$ are defined as:
\begin{equation}
S^{ij}_{\alpha}(q)=e^{i}(q)\xi^{j}_{\alpha}(q)-e^{j}(q)\xi^{i}_{\alpha}(q).
\end{equation}
\end{theorem}
\begin{proof}
Let us consider the normal coordinate representation $\mathcal{R}_{\mathbf{x}}$ of the manifold $\mathcal{M}$ centered at the point $\bar{x}$. Without lost of generality, let us suppose that this point corresponds to the origin $\mathbf{x}=0$ of the normal coordinate system $\mathcal{R}_{\mathbf{x}}$. It is convenient to adopt the notation convention $\bar{A}=A(\mathbf{x}=0|\theta)$ to simplify mathematical expressions. Besides, let us denote by $\mathbf{x}_{g}(s|\mathbf{e})$ the geodesic family derived from equations:
\begin{equation}
\frac{d x^{k}(s)}{ds}=\upsilon^{k}(\mathbf{x}|\theta),\:
\frac{d \upsilon^{k}(\mathbf{x}|\theta)}{ds}+\Gamma^{k}_{ij}(\mathbf{x}|\theta)\upsilon^{i}(\mathbf{x}|\theta)\upsilon^{j}(\mathbf{x}|\theta)=0,
\end{equation}
where the vector $\mathbf{e}=\left\{e^{k}(q)\right\}$ are the components of the tangent vector $\upsilon^{k}(\mathbf{x}|\theta)$ at the origin, $e^{k}(q)=\upsilon^{k}(0|\theta)$. This geodesic family can be expressed in terms of power-series of the arc-length parameter $s$ as follows:
\begin{equation}
x^{i}_{g}(s|\mathbf{e})=e^{i}(q)s-\frac{1}{6}s^{3}\partial_{l}\bar{\Gamma}^{i}_{jk}e^{j}(q)e^{k}(q)e^{l}(q)+O(s^{3}),
\end{equation}
where $\partial_{l}\bar{\Gamma}^{i}_{jk}=\partial \bar{\Gamma}^{i}_{jk}(0|\theta)/\partial x^{l}$ is the partial derivative of the affine connection at the origin:
\begin{equation}
\partial_{l}\bar{\Gamma}^{i}_{jk}=\frac{1}{2}\bar{g}^{im}\frac{\partial}{\partial x^{l}}
\left[\frac{\partial g_{mk}(0|\theta)}{\partial x^{j}}+\frac{\partial g_{mj}(0|\theta)}{\partial x^{k}}-\frac{\partial g_{kl}(0|\theta)}{\partial x^{m}}\right].
\end{equation}
Using the simplified expression of the curvature tensor in normal coordinates (\ref{second.curvature}), one can obtain the components of the projected metric tensor $g_{\alpha\beta}(\ell,q)$ on the boundary $\partial \mathbb{S}^{(n)}(\bar{X},\ell)$:
\begin{equation}
g_{\alpha\beta}(\ell,q)=\ell^{2}\kappa_{\alpha\beta}(q)-
\frac{1}{12}\ell^{4}\bar{R}_{ijkl}S^{ij}_{\alpha}(q)S^{kl}_{\beta}(q)+O(\ell^{4}),
\end{equation}
where $\xi^{i}_{\alpha}(q)$ are the quantities defined by equation (\ref{xi}). This last approximation leads to the asymptotic distribution (\ref{asymptotic}).
\end{proof}
\begin{remark}\label{irred.coup}
For a statistical manifold $\mathcal{M}$ with dimension $n=\mathrm{dim}(\mathcal{M})>2$, the spherical function $\mathcal{F}(q|\theta)$ characterizes the local anisotropy of the distribution function (\ref{polar}) at the neighborhood of the origin $\ell=0$, as well as the \textbf{irreducible statistical coupling} among the radial coordinate $\ell$ and the spherical coordinates $q$. The case of the curved two-dimensional statistical manifold $\mathcal{M}$ is special because of the spherical function takes the constant value $\mathcal{F}(q|\theta)\equiv 2\bar{R}$, where $\bar{R}$ is the curvature scalar at $\ell=0$. This results implies the local isotropic character of the spherical coordinate representation (\ref{polar}) at $\ell=0$ for any two-dimensional statistical manifold $\mathcal{M}$ as well as the existence of an \textbf{apparent statistical decoupling} between the radial coordinate $\ell$ and the spherical coordinate $q$ for $\ell$ sufficiently small.
\end{remark}
\begin{proof}
For the sake of convenience, let us consider a normal coordinate representation $\mathcal{R}_{\mathbf{x}}$ centered at the point $\bar{x}$. Moreover, let us employ the usual spherical coordinates $q=(q_{1},q_{2},\ldots q_{n-1})$ that parameterize the hyper-surface of a $(n-1)$-dimensional Euclidean sphere $\mathbb{S}^{(n-1)}(\bar{x}|\ell)$ with small radius $\ell$. Hereafter, let us introduce the notation convention $\alpha=(\dot{1}, \dot{2}, \ldots)$ to distinguish between the Greek indexes $(\alpha,\beta)$ and the Latin indexes $(i,j,k,l)$. The simplest case corresponds to the two-dimensional statistical manifold $\mathcal{M}$, where the vectors $\mathbf{e}(q)$ and $\mathbf{\xi}_{\dot{1}}(q)$ are given by:
\begin{equation}
\mathbf{e}(q)=(\cos q_{1},\sin q_{1})\rightarrow \xi_{\dot{1}}(q)=(-\sin q_{1},\cos q_{1}).
\end{equation}
Here, the values of spherical coordinate $q_{1}$ belong to the interval $0\leq q_{1}< 2\pi$. The previous vectors lead to $S^{12}_{\dot{1}}(q)=1$ and $\kappa_{\dot{1}\dot{1}}(q)=1$.  The only non-vanishing independent component of the curvature tensor $\bar{R}_{ijkl}$ is $\bar{R}_{1212}$. Thus, the spherical function $\mathcal{F}(q|\theta)$ can be expressed as follows $\mathcal{F}(q|\theta)\equiv 4\bar{R}_{1212}\equiv 2\bar{R}$. This results implies that the asymptotic distribution (\ref{asymptotic}) is isotropic for any two-dimensional statistical manifold $\mathcal{M}$, thus describing an apparent statistical decoupling among the radial coordinate $\ell$ and the spherical coordinate $q_{1}$ for $\ell$ sufficiently small. Such a statistical decoupling is fictitious because of the points $(\ell,q)$ with $\ell=0$ actually corresponds to the same point of the statistical manifold $\mathcal{M}$, so that, the radial coordinate $\ell$ and the spherical coordinates $q$ are not independent in the neighborhood of the origin $\ell=0$.

The first case with larger dimensionality is the 3-dimensional irreducible statistical manifold $\mathcal{M}$, where the quantities $\mathbf{e}(q)$, $\mathbf{\xi}_{\dot{1}}(q)$ and $\mathbf{\xi}_{\dot{2}}(q)$ are given by:
\begin{eqnarray}
\mathbf{e}(q)&=&(\cos q_{1}\cos q_{2},\cos q_{1}\sin q_{2},\sin q_{1}),\\ \nonumber
\mathbf{\xi}_{\dot{1}}(q)&=&(-\sin q_{1}\cos q_{2},-\sin q_{1}\sin q_{2},\cos q_{1}), \\
\mathbf{\xi}_{\dot{2}}(q)&=&(-\cos q_{1}\sin q_{2},\cos q_{1}\cos q_{2},0).\nonumber
\nonumber
\end{eqnarray}
Here, the admissible values of the spherical coordinates $(q_{1},q_{2})$ now belong to the interval $-\pi/2\leq q_{1}<\pi/2$ and $-\pi\leq q_{2}< \pi$. The non-vanishing components $\kappa _{\alpha \beta}(q)$ are:
\begin{equation}
\kappa _{\dot{1}\dot{1}}(q)=1\mbox{ and }\kappa _{\dot{2}\dot{2}}(q)=\cos^{2}q_{1},
\end{equation}
while the quantities $S_{\alpha }^{ij}(q)$ are given by:
\begin{eqnarray}\nonumber
&S_{\dot{1}}^{12}(q) =0,~S_{\dot{2}}^{12}=\cos ^{2}q_{1},  S_{\dot{1}}^{23}(q) =\sin q_{2},~S_{\dot{2}}^{23}(q)=-\sin q_{1}\cos q_{1}\cos
q_{2}, &\\
&S_{\dot{1}}^{31}(q) =-\cos q_{2},~S_{\dot{2}}^{31}(q)=-\sin q_{1}\cos q_{1}\sin
q_{2}. &
\end{eqnarray}
Let us introduce the \emph{anisotropic functions} $G^{ijkl}(q)$:
\begin{equation}
G^{ijkl}(q)=\kappa ^{\alpha \beta }(q)S_{\alpha }^{ij}(q)S_{\beta }^{kl}(q),
\end{equation}
which exhibit the same properties of the curvature tensor $\bar{R}_{ijkl}$ under the permutation of indexes. The only non-vanishing independent components of the curvature tensor $\bar{R}_{ijkl}$ are $\left( \bar{R}_{1212},\bar{R}_{2323},\bar{R}_{3131}\right) $ and $\left(\bar{R}_{1223},\bar{R}_{2331},\bar{R}_{3112}\right) $. A simple calculation yields the following results:
\begin{eqnarray}
&G^{1212}(q) =\cos ^{2}q_{1},\,G^{2323}(q) =\sin ^{2}q_{2}+\sin ^{2}q_{1}\cos ^{2}q_{2},\nonumber
\\
&G^{3131}(q) =\cos ^{2}q_{2}+\sin ^{2}q_{1}\sin ^{2}q_{2},
\\
& G^{1223}(q) =-\sin q_{1}\cos q_{1}\cos q_{2},\, G^{3112}(q) =-\sin q_{1}\cos q_{1}\sin q_{2}, \nonumber
\\
&G^{2331} (q) =-\cos q_{2}\sin q_{2}+\sin ^{2}q_{1}\cos q_{2}\sin q_{2}, \nonumber
\end{eqnarray}
The spherical function $\mathcal{F}(q|\theta)$ can be finally expressed as follows:
\begin{eqnarray}\nonumber
\mathcal{F}\left( q|\theta\right) =4\left[ \bar{R}_{1212}G^{1212}(q)+\bar{R}_{2323}G^{2323}(q)+\bar{R}_{3131}G^{3131}(q)+\right.\\
+\left.\bar{R}_{1223}G^{1223}(q)+\bar{R}_{2331}G^{2331}(q)+\bar{R}_{3112}G^{3112}(q)\right] ,\label{expansion}
\end{eqnarray}
which describes an anisotropic character of the spherical coordinate representation (\ref{polar}) for $\ell$ sufficiently small. Such an anisotropic guarantees the coupling between the radial coordinate $\ell$ and the spherical coordinates $q=(q_{1},q_{2})$. This type of coupling exhibits an \emph{irreducible character} because of the consideration of local coordinate change does not affect a scalar function as the spherical function $\mathcal{F}\left( q|\theta\right)$. In general, the anisotropic character of the spherical function $\mathcal{F}\left( q|\theta\right)$ will be observed for any $n$-dimensional irreducible statistical manifold $\mathcal{M}$ with $n>2$.
\end{proof}
\begin{corollary}
The statistical curvature tensor $R_{ijkl}(x|\theta)$ allows to introduce some local and global invariant measures to characterize both the intrinsic curvature of the manifold $\mathcal{M}$ as well as the existence of an irreducible statistical dependence among the stochastic variables $x$. They are the \textit{curvature scalar} $R(x|\theta)$ introduced in equation (\ref{scalar.curvature}), the
\textit{spherical curvature scalar} $\Pi(\ell,q|\theta)$:
\begin{equation}
\Pi(\ell,q|\theta)=g^{\alpha\beta}(\ell,q|\theta)R_{ijkl}(\ell,q|\theta)S^{ij}_{\alpha}
(\ell,q|\theta)S^{kl}_{\alpha}(\ell,q|\theta)
\end{equation}
with $S^{ij}_{\alpha}(\ell,q|\theta)$ being:
\begin{equation}
X^{ij}_{\alpha}(\ell,q|\theta)=\upsilon^{i}(\ell,q|\theta)\tau^{j}_{\alpha}(\ell,q|\theta)-\upsilon^{j}(\ell,q|\theta)
\tau^{i}_{\alpha}(\ell,q|\theta),
\end{equation}
which arises as a local measure of the coupling between the radial $\ell$ and the spherical coordinates $q$ in the spherical representation of the distribution function (\ref{polar}), and finally, the \textit{gaussian potential}
$\mathcal{P}(\theta)=-\log \mathcal{Z}(\theta)$, which arises as a global invariant measure of the curvature of the manifold $\mathcal{M}$.
\end{corollary}
\begin{proof}
The curvature scalar $R(x|\theta)$ is the only invariant associated with the first and second partial derivatives of the metric tensor $g_{ij}(x|\theta)$. The consideration of the spherical representation of the distribution function (\ref{polar}) allows to introduce the normal $\upsilon^{i}(\ell,q|\theta)$ and tangential vectors $\tau^{i}_{\alpha}(\ell,q|\theta)$, as well as the projected metric tensor $g_{\alpha\beta}(\ell,q|\theta)=g_{ij}(\ell,q|\theta)\tau^{i}_{\alpha}(\ell,q|\theta)\tau^{j}_{\beta}(\ell,q|\theta)$
associated with the constant information potential hyper-surface $\mathbb{S}^{(n-1)}(\bar{x}|\ell)$. This framework leads to introduce the spherical curvature scalar $\Pi(\ell,q|\theta)$ as a direct generalization of the spherical function $\mathcal{F}(q|\theta)$ of the asymptotic distribution function (\ref{asymptotic}). The role of the gaussian potential $\mathcal{P}(\theta)$ as a global invariant measure of the curvature of the manifold
$\mathcal{M}$ can be easily evidenced starting from the spherical representation of the distribution function (\ref{polar}). Integrating over the spherical coordinates $q$, one obtains the following expression for the gaussian partition function:
\begin{equation}\label{normalization}
\mathcal{Z}(\theta)=\frac{1}{\sqrt{2\pi}}\int^{+\infty}_{0}\exp\left(-\frac{1}{2}\ell^{2}\right)\Sigma_{g}(\ell|\theta)d\ell,
\end{equation}
where $\Sigma_{g}(\ell|\theta)$ denotes the area of the constant information potential hyper-surface $\mathbb{S}^{(n-1)}(\bar{x}|\ell)$ normalized by the factor $(2\pi)^{(n-1)/2}$. For the special case of the $n$-dimensional Euclidean real space $\mathbb{R}^{n}$, the quantity $\Sigma_{flat}(\ell|\theta)$ is given by:
\begin{equation}\label{flat.area}
\Sigma_{flat}(\ell|\theta)=\frac{\sqrt{\pi}\ell^{n-1}}{2^{\frac{n-1}{2}}\Gamma\left(\frac{n}{2}\right)}.
\end{equation}
Equation (\ref{normalization}) can be rewritten as follows:
\begin{equation}\label{Zflat}
\mathcal{Z}(\theta)=1+\frac{1}{\sqrt{2\pi}}\int^{+\infty}_{0}\exp\left(-\frac{1}{2}\ell^{2}\right)\sigma(\ell|\theta)\Sigma_{flat}(\ell|\theta)d\ell,
\end{equation}
where $\sigma(\ell|\theta)$ represents \textbf{spherical distortion}:
\begin{equation}\label{sphericaldistortion}
\sigma(\ell|\theta)=\Sigma_{g}(\ell|\theta)/\Sigma_{flat}(\ell|\theta)-1
\end{equation}
that characterizes how much differ the area of the sphere $\mathbb{S}^{(n-1)}(\bar{x},\ell)\subset\mathcal{M}$ due to its intrinsic curvature. Since the gaussian partition function $\mathcal{Z}(\theta)=1$ for the case of the $n$-dimensional Euclidean real space $\mathbb{R}^{n}$, a non-vanishing gaussian potential $\mathcal{P}(\theta)$ appears as a global invariant measure of the intrinsic curvature of the statistical manifold $\mathcal{M}$, and hence, as a global indicator of the existence of irreducible statistical correlations.
\end{proof}
\begin{definition}
The value the scalar curvature $R(\bar{x}|\theta)$ at the point $\bar{x}$ with maximum information potential $\mathcal{S}(x|\theta)$ allows to introduce the \textbf{curvature radius} $\ell_{c}$:
\begin{equation}\label{ellc}
R(\bar{x}|\theta)=\frac{1}{\ell^{2}_{c}},
\end{equation}
which represents the statistical distance where distortion of Euclidean geometry is appreciable, and hence, where the statistical correlations among the coordinates $x=(x^{1},x^{2},\ldots x^{n})$ turn irreducible.
\end{definition}
\begin{theorem}
If the curvature radius $\ell_{c}$ is sufficiently large, the gaussian potential $\mathcal{P}(\theta)$ can be estimated as follows:
\begin{equation}  \label{eggrenium}
\mathcal{P}(\theta)\simeq\frac{1}{6}R(\bar{x}|\theta).
\end{equation}
\end{theorem}
\begin{proof}
According to spherical representation of the Euclidean gaussian distribution (\ref{SCRGaussianFamily}), the expectation value of the radius $\ell$ is $\left\langle\ell^{2}\right\rangle\equiv n$ in this approximation level. This result implies that that gaussian distribution (\ref{SCRGaussianFamily}) differs in a significant way from zero in a small region of radius $\sqrt{n}$. Therefore, the Euclidean gaussian distribution arises as a good approximation when the curvature radius $\ell_{c}$ is sufficiently large. The approximation formula (\ref{AF1}) allows to express the spherical distortion (\ref{sphericaldistortion}) as follows:
\begin{equation}
\sigma(\ell|\theta)=-\frac{R(\bar{x}|\theta)}{6n}\ell^{2}+O(\ell^{4}).
\end{equation}
The estimation (\ref{eggrenium}) is directly obtained from the integration formula (\ref{Zflat}). The applicability of this estimation requires the condition $\ell_{c}\gg 1$, which guarantees that the correction term associated with the spherical function in equation (\ref{asymptotic}) is very small.
\end{proof}
\begin{corollary}
A general criterium for the applicability of the gaussian approximation of the given distributions family (\ref{DF}) is the following:
\begin{equation}\label{criterium}
\ell_{c}\gg 1,
\end{equation}
where $\ell_{c}$ represents the curvature radius (\ref{ellc}).
\end{corollary}

\subsection{A simple illustration example}

The major problem of fluctuation geometry is the derivation of the metric tensor $g_{ij}(x|\theta)$ for a given continuous distributions family $dp(x|\theta)$. An amenable treatment of problem (\ref{cov.equation}) is possible for some particular cases, overall, when some type of symmetry is present. This is the case of distributions family discussed in this subsection:
\begin{equation}
dp\left( r,\varphi |\theta \right) =\frac{1}{\mathcal{A}\left( \theta \right) }\exp %
\left[ -\frac{1}{2}r^{2}\right] \frac{rdrd\varphi }{\sqrt{\theta ^{2}+r^{2}}}.
\label{ex.1}
\end{equation}%
The same one is defined on a two-dimensional statistical manifold $\mathcal{M}$ that is expressed in a polar coordinate representation $\mathcal{R}_{\rho}$ with $\rho=(r,\varphi)$, where $0\leq r<+\infty $ and $0\leq \varphi <2\pi $. The normalization function $\mathcal{A}\left( a\right) $ is given by:
\begin{equation}
\mathcal{A}\left( \theta \right) =e^{\frac{1}{2}\theta ^{2}}\pi \sqrt{2\pi }\mathrm{erfc}\left( \frac{\theta }{\sqrt{2}}\right),
\end{equation}
where $\mathrm{erfc}(x)$ is the \emph{complementary error function}:
\begin{equation}\label{erfc}
\mathrm{erfc}(x)=\frac{2}{\sqrt{\pi}}\int_{x}^{+\infty}e^{-z^{2}}dz.
\end{equation}
This distributions family exhibits an axial symmetry in this coordinate representation. It is easy to check that the same one can be expressed into the Riemannian gaussian representation (\ref{UGR}) considering the distance notion:
\begin{equation}\label{ds.ex1}
ds^{2}=dr^{2}+\frac{\theta ^{2}r^{2}}{\theta ^{2}+r^{2}}d\varphi ^{2},
\end{equation}
which is centered at the point $r=0$. Thus, the separation distance $\ell^{2}(r,\varphi)\equiv r^{2}$ and the information potential $\mathcal{S}\left( r,\varphi |\theta \right) $ is given by:
\begin{equation}
\mathcal{S}\left( r,\varphi |\theta\right) =\mathcal{P}\left( \theta\right) -\frac{1}{2}r^{2},
\end{equation}
where $\mathcal{P}\left( \theta \right) $ is the gaussian potential obtained from the gaussian partition function $\mathcal{Z}\left( \theta \right)$:
\begin{equation}
\mathcal{Z}\left( \theta \right) =\frac{1}{2\pi }\theta\mathcal{A}\left( \theta \right)=\sqrt{\pi}e^{\frac{1}{2}\theta ^{2}}\frac{\theta}{\sqrt{2} }\mathrm{erfc}\left( \frac{\theta }{\sqrt{2}}\right) .
\end{equation}
The probability weight $\omega(r,\varphi|\theta)$ and the curvature scalar $R(r,\varphi|\theta)$ associated with the distributions family (\ref{ex.1}) are given by:
\begin{equation}
\omega(r,\varphi|\theta)=\frac{1 }{\mathcal{Z}\left( \theta
\right) }\exp \left[ -\frac{1}{2}r^{2}\right]\mbox{ and }R(r,\varphi|\theta)=\frac{6\theta^{2}}{(\theta^{2}+r^{2})^{2}}.
\end{equation}
Apparently, the distributions family (\ref{ex.1}) can be decomposed into two independent distributions:
\begin{equation}
dp^{\left( 1\right) }\left( r|\theta \right) =\frac{1 }{\mathcal{Z}\left( \theta
\right) }\exp \left[ -\frac{1}{2}r^{2}\right] \frac{rdr}{\sqrt{\theta
^{2}+r^{2}}}\mbox{ and }dp^{\left( 2\right) }\left( \varphi \right) =\frac{1}{2\pi }d\varphi.
\end{equation}
However, such a ``statistical independence'' between the variables $r$ and $\varphi$ is \textit{fictitious decoupling} because of the points $\left( r,\varphi \right) $ with $r=0$ actually correspond to the same point in the statistical manifold $\mathcal{M}$ without mattering about the value of the angle variable $\varphi $. Such an apparent decomposition is a consequence of the non-bijective character of coordinate representation of the manifold $\mathcal{M}$ in terms of polar coordinates $\rho=(r,\varphi)$, which disappears if one considers any coordinate representation of the statistical manifold $\mathcal{M}$. A simple case is the coordinate representation $\mathcal{R}_{\mathbf{x}}$, where $\mathbf{x}=(x,y)$ denotes the cartesian coordinates $x=r\cos\varphi$ and $y=r\sin\varphi$. Thus, the distance notion (\ref{ds.ex1}) can be rewritten as follows:
\begin{equation}
ds^{2}=\frac{x^{2}+\theta ^{2}}{\theta ^{2}+x^{2}+y^{2}}dx^{2}+\frac{xy}{%
\theta ^{2}+x^{2}+y^{2}}2dxdy+\frac{y^{2}+\theta ^{2}}{\theta ^{2}+x^{2}+y^{2}%
}dy^{2},
\end{equation}
while the distributions family (\ref{ex.1}) adopts the following form:
\begin{equation}
dp(x,y|\theta)=\frac{1 }{\mathcal{Z}\left( \theta
\right) }\exp \left[ -\frac{1}{2}(x^{2}+y^{2})\right]\frac{\theta dx dy}{2\pi\sqrt{x^{2}+y^{2}+\theta^{2}}}.
\end{equation}
The cartesian coordinates $\mathbf{x}=(x,y)$ can be regarded as \emph{normal coordinates} at the origin point $(0,0)$, since $g_{ij}(0,0|\theta)=\delta_{ij}$ and $\partial g_{ij}(0,0|\theta)/\partial x^{k}=0$, where $x^{1}=x$ and $x^{2}=y$. Although the small neighborhood of the point $(0,0)$ looks-like a small Euclidean subset, the coordinates $x$ and $y$ exhibits an irreducible coupling due to the presence of the dividend $\sqrt{x^{2}+y^{2}+\theta^{2}}$. The statistical manifold $\mathcal{M}$ exhibits the same differential structure of the two-dimensional real space $\mathbb{R}^{2}$, but its Riemannian structure is different because of $\mathcal{M}$ is a curved manifold. Consequently, distributions family (\ref{ex.1}) exhibits an irreducible statistical dependence.

\begin{figure}
\begin{center}
  \includegraphics[width=5.0in]{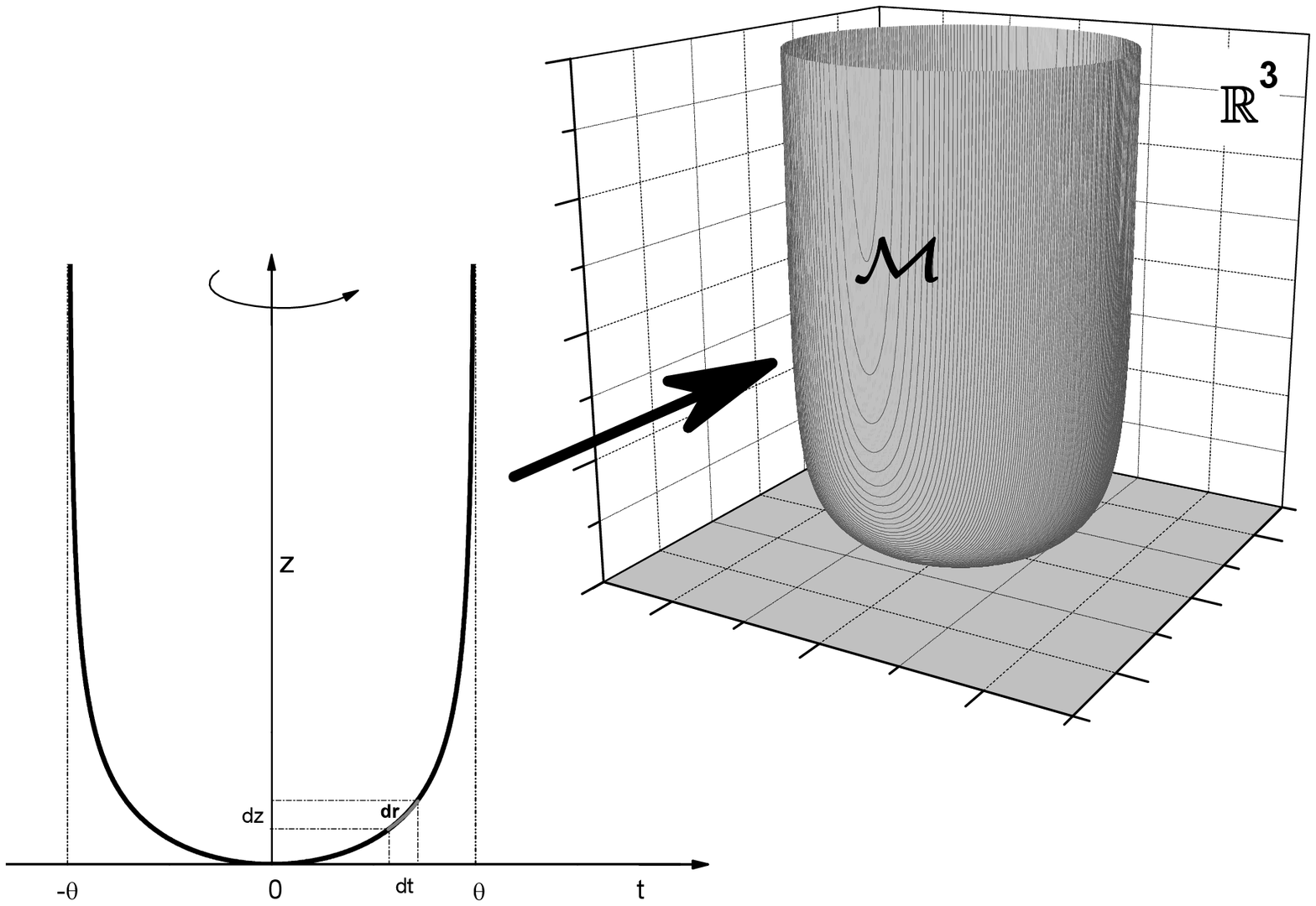}
\end{center}
\caption{The geometry of the statistical manifold $\mathcal{M}$ associated with the distributions family (\ref{ex.1}) is fully equivalent to curved geometry defined on the revolution surface obtained from the dependence $z=z(t)$, which is embedded in the $3$-dimensional real space $\mathbb{R}^{3}$. As expected, this manifold cannot be decomposed into independent manifolds.}\label{surface.eps}
\end{figure}

\begin{figure}
\begin{center}
  \includegraphics[width=4.5in]{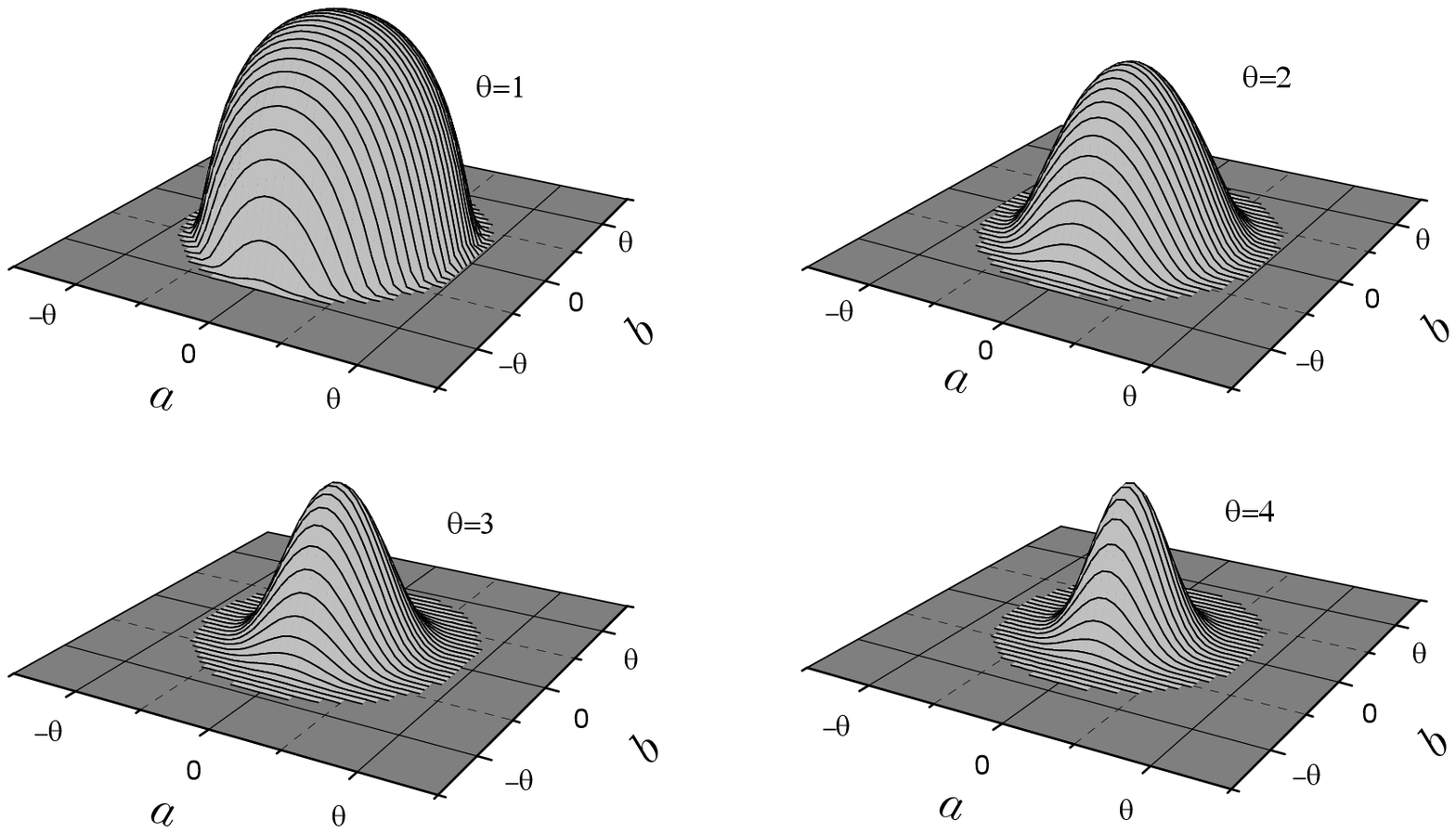}
\end{center}
\caption{Behavior of the probability weight $\omega(a,b|\theta)$ for some values of the control parameter $\theta$. Here, the variables $(a,b)$ are the cartesian coordinates $a=t \cos\varphi$ and $b=t\sin\varphi$. The probability weight $\omega(a,b|\theta)$  behaves as a usual gaussian distribution function when the control parameter $\theta$ is sufficiently large.}\label{weight.eps}
\end{figure}

\begin{figure}
\begin{center}
  \includegraphics[width=3.5in]{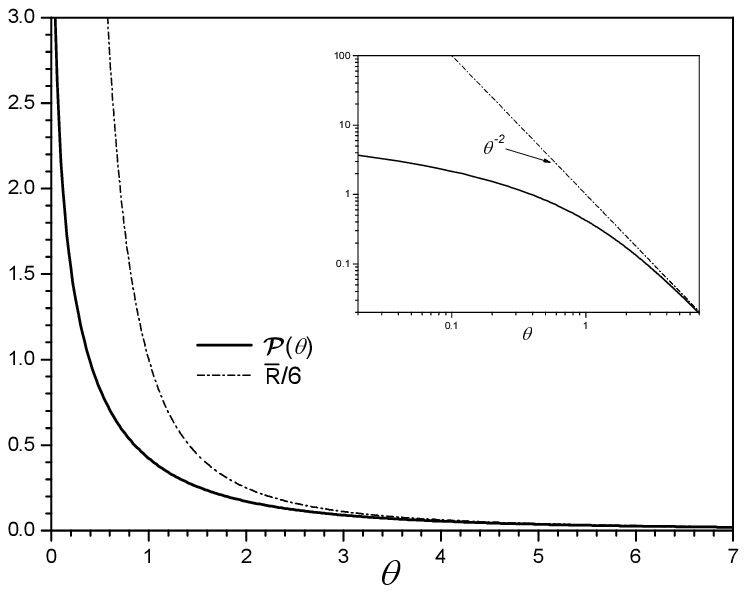}
\end{center}
\caption{Comparison between the gaussian potential $\mathcal{P}(\theta)$ and the sixth part of the central value of the curvature scalar $\bar{R}=R(r=0,\varphi|\theta)=6/\theta^{2}$. As expected, there exist a convergence of these functions for $\theta$ sufficiently large. Inset panel: The same dependencies using a log-log scale to illustrate the asymptotic dependence $1/\theta^{2}$ of the gaussian potential $\mathcal{P}(\theta)$ for large values of control parameter $\theta$.}\label{convergence.eps}
\end{figure}

A visual representation for the statistical manifold $\mathcal{M}$ can be obtained considering the coordinate change $\phi:\mathcal{R}_{\rho}\rightarrow \mathcal{R}_{\tau}$ with $\tau=(t,\varphi)$, which only involves a change in the radial coordinates $r=\theta^{2}t /\sqrt{\theta^{2}-t^{2}}$. The distance notion (\ref{ds.ex1}) is rewritten as:
\begin{equation}
ds^{2}=\frac{\theta^{6}}{(\theta^{2}-t^{2})^{3}}dt^{2}+t^{2}d\varphi^{2},
\end{equation}
while the distributions family:
\begin{equation}
dp(t,\varphi|\theta)=\frac{1 }{\mathcal{Z}\left( \theta
\right) }\exp \left[ -\frac{1}{2}\ell^{2}(t,\varphi)\right]\frac{\theta^{3} t dt d\varphi}{2\pi\sqrt{(\theta^{2}-t^{2})^{3}}},
\end{equation}
where $\ell^{2}(t,\varphi)=\theta^{4}t^{2}/(\theta^{2}-t^{2})$. The points on the one-dimensional sphere $\mathbb{S}^{(1)}$ defined by the curve $t=\theta$ are infinitely separated from the origin point with $t=0$. This region now appears as the boundary of the statistical manifold $\mathcal{M}$ in the coordinate representation $\mathcal{R}_{\tau}$. The radial coordinate $r$ can be regarded as the arc-length of the curve $z(t)$  defined in the plane $(t,z)$ of two-dimensional real space $\mathbb{R}^{2}$. This assumption allows to express the curve $z(t)$ as follows:
\begin{equation}
dz=\sqrt{dr^{2}-dt^{2}}\rightarrow z(t)=\theta f (t/\theta),
\end{equation}
where $f(x)$ is defined as:
\begin{equation}
f(x)=\int^{x}_{0}\sqrt{\frac{1-(1-\zeta^{2})^{3}}{(1-\zeta^{2})^{3}}}d\zeta.
\end{equation}
The rotation of this curve around the axis $z$ generates the \emph{revolution surface} represented in figure \ref{surface.eps}, which is defined in the $3$-dimensional real space $\mathbb{R}^{3}$. Riemannian geometry of the statistical manifold $\mathcal{M}$ is fully equivalent to the curved geometry defined on this revolution surface. For $r$ sufficiently large, the local geometry defined on this surface asymptotically behaves as the Euclidean geometry defined on surface of a cylinder $\mathbb{C}^{(2)}$ with radius $t=\theta$. The cylinder $\mathbb{C}^{(2)}$ is a Euclidean manifold that can be decomposed into the one-dimensional sphere $\mathbb{S}^{(1)}$ and the one-dimensional real space $\mathbb{R}$, $\mathbb{C}^{(2)}=\mathbb{S}^{(1)}\bigotimes \mathbb{R}$. On the other hand, the small neighborhood at the point $t=0$ locally behaves as a small subset of two-dimensional sphere $\mathbb{S}^{(2)}$ with curvature radius $\ell_{c}=1/\sqrt{\bar{R}}=\theta/\sqrt{6}$. Thus, the statistical manifold $\mathcal{M}$ drops to the two-dimensional real space $\mathbb{R}^{2}$ when $\theta\rightarrow \infty$.

The behavior of the probability weight $\omega(a,b|\theta)$ for some values of the control parameter $\theta$ is illustrated in figure \ref{weight.eps}. Here, the variables $(a,b)$ are the cartesian coordinates $a=t \cos\varphi$ and $b=t\sin\varphi$. As expected, the curved character of the manifold $\mathcal{M}$ is significant for small values of the control parameter $\theta$, which manifests in the non-gaussian character of the probability weight $\omega(a,b|\theta)$. Conversely, this function asymptotically behaves as a gaussian distribution for large values of the control parameter $\theta$. The applicability of the estimation formula (\ref{eggrenium}) in this region is clearly evidenced in figure \ref{convergence.eps}, where one observes the convergence of the gaussian potential $\mathcal{P}(\theta)$ and sixth part of the central value of the curvature scalar $\bar{R}/6=1/\theta^{2}$. The curvature radius $\ell_{c}$ for this example is $\ell_{c}=\theta/\sqrt{6}$. Considering the general criterium (\ref{criterium}) for the applicability of estimation formula (\ref{eggrenium}) and the gaussian approximation, one obtains $\ell_{c}\gg 1\rightarrow \theta\gg 2.45$, which is in a good agreement with the convergence observed in figure \ref{convergence.eps}.

\section{Riemannian extension of Einstein's fluctuation theory}\label{Implications}

Inference geometry and fluctuation geometry can be applied to any physical theory with a statistical formulation, such as statistical mechanics and quantum mechanics. In particular, they can be employed to analyze the geometric features of continuous distributions (\ref{BG}) and (\ref{QM.dist}). Inference geometry has been employed in statistical mechanics to study phase transitions \cite{Janke}-\cite{Janyszek}, as well as in the context of thermodynamics geometry \cite{Crooks}. Moreover, inference theory and its geometry have been adapted to the mathematical apparatus of quantum mechanics \cite{Brody1}-\cite{Wootters}. Until now, applications of fluctuation geometry are only focussed on classical statistical mechanics, specifically, in the framework of Riemannian extension of Einstein's fluctuation theory \cite{Vel.GEO,Vel.FG}. Needless to say that potential applications of fluctuation geometry to quantum mechanics represent an attractive field for future developments.

Fluctuation geometry naturally arises as the mathematical apparatus of a Riemannian extension of Einstein fluctuation theory \cite{Vel.GEO}. The term \emph{extension} clarifies that this approach is not a simple application of fluctuation geometry on this physical theory. On the contrary, the existence of fluctuation geometry inspires a re-examination of foundations of Einstein fluctuation theory based on the notions of Riemannian geometry. In this section, let us firstly review the physical foundations and the direct consequences of this geometric development. Afterwards, let us proceed to obtain certain fluctuation theorems and asymptotic formulae based on the second-order geometric expansion of fluctuation geometry.

\subsection{Einstein postulate revisited}
Let us redefine the information potential $\mathcal{S}(x|\theta)$ and the invariant volume element $d\mu(x|\theta)$ using Boltzmann constant $k$ as follows:
\begin{equation}\label{kB.rephrasing}
\mathcal{S}(x|\theta)=k\log\omega(x|\theta)\mbox{ and }d\mu(x|\theta)=\sqrt{\left|\frac{g_{ij}(x|\theta)}{2\pi k}\right|}dx.
\end{equation}
Thus, the invariant form (\ref{DF2}) of the distributions family (\ref{DF}) can be rewritten as follows:
\begin{equation}\label{cov.EP}
dp(x|\theta)=\exp\left[\mathcal{S}(x|\theta)/k\right]d\mu(x|\theta).
\end{equation}
This expression represents a \emph{covariant extension of Einstein's postulate} of classical fluctuation theory \cite{Reichl}, where the information potential $\mathcal{S}(x|\theta)$ has been identified with the \emph{thermodynamic entropy of closed system} (up to the precision of an additive constant). Hereinafter, the coordinates $x=(x^{1}, x^{2},\ldots,x^{n})$ are the relevant macroscopic observables of the closed system, e.g.: the internal energy $U$, the volume $V$, the total angular momentum $\mathbf{M}$, the magnetization $\mathcal{M}$, etc. Moreover, $\theta$ represents the set of control parameters of the given situation of thermodynamic equilibrium. The metric tensor $g_{ij}(x|\theta)$ of fluctuation geometry:
\begin{equation}\label{FG.RFT}
g_{ij}(x|\theta)=-D_{i}D_{j}\mathcal{S}(x|\theta)=-\partial_{i}\partial_{i}\mathcal{S}(x|\theta)+\Gamma^{k}_{ij}(x|\theta)\partial_{k}\mathcal{S}(x|\theta)
\end{equation}
establishes a constraint between the entropy $\mathcal{S}(x|\theta)$ and the metric tensor $g_{ij}(x|\theta)$ of the abstract manifold $\mathcal{M}$ of macroscopic observables $x$. Relations (\ref{cov.EP}) and (\ref{FG.RFT}) were early proposed in Ref.\cite{Vel.GEO}. Let us now refer to the physical foundations that justify their introduction.

According to the analogy between classical statistical mechanics and quantum mechanics \cite{Vel.CSM}, the thermodynamic entropy $\mathcal{S}(x|\theta)$ appears as a counterpart of classical action $S(\mathbf{q},t)$. For any physical theory with a geometric formulation, the classical action $S(\mathbf{q},t)$ is invariant function under certain symmetric transformations, e.g.: the general symmetries of space-time. By analogy, the thermodynamic entropy $\mathcal{S}(x|\theta)$ should exhibit similar symmetric properties. As already assumed in this approach, the state of a closed system is associated with a point $x$ in the abstract manifold $\mathcal{M}$ of macroscopic observables. Although one has to chose a coordinate representation $\mathcal{R}_{x}$ to describe the manifold $\mathcal{M}$, the physical properties of the closed system should not depend on this choice. Noteworthy that this property represents a sort of \emph{relativity principle} for classical statistical mechanics, which is identified with the \emph{requirement of general covariance} of fundamental laws of physics. In thermodynamics, the entropy $\mathcal{S}(x|\theta)$ is a \emph{state function}\footnote{State function: a property of a system that depends only on the current state of the system, not on the way in which the system acquired that state.}, and hence, the same one should behave as a \emph{scalar function}:
\begin{equation}\label{scalar}
\check{\mathcal{S}}(\check{x}|\theta)=\mathcal{S}(x|\theta)
\end{equation}
under any coordinate change $\phi:\mathcal{R}_{x}\rightarrow\mathcal{R}_{\check{x}}$. According to the original mathematical form of Einstein's postulate \cite{Reichl}:
\begin{equation}\label{EP}
dp(x|\theta)=A\exp\left[S(x|\theta)/k\right]dx,
\end{equation}
the entropy $S(x|\theta)$ must obey the following transformation rule:
\begin{equation}
\check{S}(\check{x}|\theta)=S(x|\theta)-k\log\left|\partial \check{x}/\partial x\right|,
\end{equation}
with $\left|\partial \check{x}/\partial x\right|$ being the Jacobian of the coordinate change. While expression (\ref{EP}) is incompatible with the scalar character of the entropy (\ref{scalar}), this requirement is satisfied by the covariant extension (\ref{cov.EP}). However, this generalization implies that fluctuating behavior will also depend on the metric tensor $g_{ij}(x|\theta)$. Hypothesis (\ref{FG.RFT}) establishes a constraint between the metric tensor $g_{ij}(x|\theta)$ and the entropy $\mathcal{S}(x|\theta)$. Thus, the knowledge of the entropy $\mathcal{S}(x|\theta)$ fully determines the fluctuating behavior of the closed system. The introduction of this second hypothesis is not arbitrary \cite{Vel.GEO}. The metric tensor definition (\ref{FG.RFT}) guarantees the matching of the present formulation with \emph{Ruppeiner geometry of thermodynamics} \cite{Ruppeiner1,Ruppeiner2}. Expression (\ref{FG.RFT}) represents a convenient generalization for the \emph{thermodynamic metric tensor}:
\begin{equation}\label{thermo.tensor}
g^{R}_{ij}(\bar{x})=-\frac{\partial^{2}S(\bar{x}|\theta)}{\partial x^{i}\partial x^{j}},
\end{equation}
which is introduced in the framework of gaussian approximation of Einstein's fluctuation theory. The ordinary second partial derivatives $\partial_{i}\partial_{j}a(x)\equiv\partial^{2}a(x)/\partial x^{i}\partial x^{j}$ of a scalar function as the entropy do not represent tensorial quantities of any kind for an arbitrary point $x\in\mathcal{M}$. A remarkable exception takes place at the point $\bar{x}$ where the entropy reaches its maximum value, where these quantities behave as the components of a second-rank covariant tensor. This exception was considered by Ruppeiner \cite{Ruppeiner1} to introduce the metric tensor (\ref{thermo.tensor}). Remarkably, a more convenient metric tensor $g_{ij}(x|\theta)$ can be introduced for any point $x\in\mathcal{M}$ replacing the ordinary partial derivatives $\partial_{i}$ by the covariant differentiation $D_{i}$. Unfortunately, there is a cost to pay for this generalization: definition (\ref{FG.RFT}) actually represents a \emph{set of first-order covariant differential equations} to obtain the metric tensor $g_{ij}(x|\theta)$ from the entropy $\mathcal{S}(x|\theta)$. This problem is explicitly \emph{nonlinear} and difficult to solve in most of practical situations.

\subsection{Direct implications}

Hypotheses (\ref{cov.EP}) and (\ref{FG.RFT}) lead to a \emph{geometric reinterpretation} of macroscopic behavior of the closed system, where fluctuation geometry appears as the mathematical apparatus. For example, Riemannian gaussian representation (\ref{UGR}) of distribution (\ref{cov.EP}):
\begin{equation}\label{UGRTh}
dp(x|\theta)=\frac{1}{\mathcal{Z}(\theta)}\exp\left[-\ell^{2}_{\theta}(x,\bar{x})/2k\right]d\mu(x|\theta)
\end{equation}
constitutes \emph{an exact improvement} of gaussian approximation of Einstein's fluctuation theory \cite{Reichl}:
\begin{equation}\label{gaussian.approx}
dp(x|\theta)\simeq\exp\left[-g^{R}_{ij}(\bar{x})\Delta x^{i}\Delta x^{j}/2k\right]\sqrt{\left|g^{R}_{ij}(\bar{x})/2\pi k\right|}dx,
\end{equation}
where $\Delta x^{i}=x^{i}-\bar{x}^{i}$ and $g^{R}_{ij}(\bar{x})$ is the metric tensor of Ruppeiner geometry (\ref{thermo.tensor}). Accordingly, the distance notion $\ell^{2}_{\theta}(x,\bar{x})$ associated with the metric tensor $g_{ij}(x|\theta)$ quantifies the \emph{occurrence probability} of an spontaneous deviation of the system from the state of thermodynamic equilibrium $\bar{x}$, that is, the state with maximum entropy $\mathcal{S}(x|\theta)$. By itself, equation (\ref{UGRTh}) clarifies that \emph{the study of thermo-statistical properties of a given closed system is reduced to the analysis of geometric features of the abstract manifold} $\mathcal{M}$.

The covariant components $\eta_{i}(x|\theta)$ of the vector field defined from the entropy:
\begin{equation}
\eta_{i}(x|\theta)=D_{i}\mathcal{S}(x|\theta)\equiv\partial\mathcal{S}(x|\theta)/\partial x^{i}
\end{equation}
are hereinafter referred to as the \emph{generalized restituting forces}. In non-equilibrium thermodynamics, such forces appear in the phenomenological equations \cite{Reichl}:
\begin{equation}
\frac{d}{dt}x^{i}(t)=L^{ij}\eta_{i}\left[x(t)|\theta\right]
\end{equation}
describing the relaxation dynamics of a closed system towards the state of equilibrium $\bar{x}$, with $L^{ij}$ being the \emph{matrix of transport coefficients}. According to identity (\ref{Sdecomposition}), the generalized restituting forces $\eta_{i}(x|\theta)$ are related to the entropy $\mathcal{S}(x|\theta)$ as follows:
\begin{equation}
\mathcal{P}(\theta)=\mathcal{S}(x|\theta)+\frac{1}{2}\eta^{2}(x|\theta),
\end{equation}
where $\mathcal{P}(\theta)$ is gaussian potential expressed in units of Boltzmann's constant $k$:
\begin{equation}
\mathcal{P}(\theta)=-k\log \mathcal{Z}(\theta).
\end{equation}
Considering the vanishing of the generalized restituting forces at the equilibrium state $\bar{x}$:
\begin{equation}
\eta_{i}(\bar{x}|\theta)=0,
\end{equation}
one realizes that the gaussian potential is simply the maximum value of entropy:
\begin{equation}
\mathcal{S}(\bar{x}|\theta)\equiv\mathcal{P}(\theta).
\end{equation}
According to the identity (\ref{ident.HJeq}), the generalized restituting forces $\eta_{i}(x|\theta)$ are also related to the separation distance $\ell_{\theta}(x,\bar{x})$ as follows:
\begin{equation}\label{ident.genforsep}
\eta^{2}(x|\theta)=\ell^{2}_{\theta}(x,\bar{x}).
\end{equation}
Considering the following expression:
\begin{equation}\label{deltaS}
\delta\mathcal{S}(x|\theta)=\mathcal{S}(x|\theta)-\mathcal{S}(\bar{x}|\theta)=-\eta^{2}(x|\theta)/2\equiv-\ell^{2}_{\theta}(x,\bar{x})/2,
\end{equation}
the quantities $\eta_{i}(x|\theta)$ and $\ell_{\theta}(x,\bar{x})$ characterize the deviation of the
entropy of the closed system $\mathcal{S}(x|\theta)$ from its maximum value $\mathcal{S}(\bar{x}|\theta)$.

\textbf{Theorem 3} obtained in Ref.\cite{Vel.FG} guarantees the existence and uniqueness of equilibrium state $\bar{x}$ of the closed system. This fact is a direct consequence of the vanishing of the probability weight $\omega(x|\theta)$ on the boundary of the manifold $\partial\mathcal{M}$ and the concavity of the thermodynamic entropy $\mathcal{S}(x|\theta)$ associated with definition (\ref{FG.RFT}). The vanishing of the probability weight $\omega(x|\theta)$ is associated with \textbf{Axiom 4} of fluctuation geometry \cite{Vel.FG}, which establishes the vanishing of the probability density $\rho(x|\theta)$ at the boundary points. Noteworthy that this condition is a common feature of distribution functions in classical statistical mechanics\footnote{A typical example is the equilibrium distribution function:
 \[dp(E_{A},V_{A}|E_{T},V_{T})=C\Omega_{A}(E_{A},V_{A})\Omega_{B}(E_{T}-E_{A},V_{T}-V_{A})dE_{A}dV_{A}, \]
which corresponds to two separable short-range interacting systems $A$ and $B$ with additive total energy $E_{T}=E_{A}+E_{B}$ and volume $V_{T}=V_{A}+V_{B}$. Here, $\Omega_{A}$ and $\Omega_{B}$ are the densities of states of each system, while $C$ is the normalization constant. Noteworthy that this distribution vanishes at the boundary of the intervals $\min(E_{A})\leq E_{A}\leq E_{T}-\min(E_{B})$ and $\min(V_{A})\leq V_{A}\leq V_{T}-\min(V_{B})$ because of the density of states of classical systems vanishes as $\Omega(E,V)\propto (E-E_{min})^{\alpha}(V-V_{min})^{\gamma}$ with positive exponents $\alpha$ and $\gamma$ when the energy $E$ and volume $V$ approach their minimum values.}. According to equation (\ref{UGRTh}), the vanishing of the probability weight $\omega(x|\theta)$ at the boundary $\partial\mathcal{M}$ also implies that any boundary point $x_{b}\in\partial\mathcal{M}$ is infinitely far from the equilibrium state $\bar{x}$, $\ell_{\theta}(x_{b},\bar{x})=+\infty$.

\subsection{Invariant fluctuation theorems}

Conventionally \cite{Reichl}, results of Einstein's fluctuation theory involve expectation values such as the macroscopic observables $\left\langle x^{i}\right\rangle$, their self-correlation functions $\left\langle \delta x^{i}\delta x^{j}\right\rangle$, etc. However, these expectation values crucially depend on the coordinate representation $\mathcal{R}_{x}$ employed to describe the abstract manifold $\mathcal{M}$. In the present Riemannian approach, one is interested on the calculation of the expectation values of \emph{scalar functions} $a(x|\theta)$:
\begin{equation}
\left\langle a(x|\theta) \right\rangle=\int_{\mathcal{M}}a(x|\theta)dp(x|\theta).
\end{equation}
Fluctuation relations that involve this type of expectation values can be referred to as \emph{invariant fluctuation theorems}.

An important case of invariant fluctuation theorem is the following identity:
\begin{equation}\label{fluct.rel}
\left\langle kD_{i}w^{i}(x|\theta)+\eta_{i}(x|\theta)w^{i}(x|\theta)\right\rangle=0.
\end{equation}
Here, $w^{i}(x|\theta)$ denotes the contravariant components of a differentiable vector field $\mathbf{w}$ with a well-defined expectation value $\left\langle\mathbf{\eta}\cdot\mathbf{w}\right\rangle=\left\langle\eta_{i}(x|\theta)w^{i}(x|\theta)\right\rangle$. To proceed the demonstration of this identity, let us introduce the contravariant components $v^{i}(x|\theta)$ of the auxiliary vector field $\mathbf{v}$:
\begin{equation}\label{aux.vector}
\upsilon^{i}(x|\theta)=\frac{1}{\mathcal{Z}(\theta)}\exp\left[-\eta^{2}(x|\theta)/2k\right]w^{i}(x|\theta).
\end{equation}
Noteworthy that the factor:
\begin{equation}
\frac{1}{\mathcal{Z}(\theta)}\exp\left[-\eta^{2}(x|\theta)/2k\right]\equiv\frac{1}{\mathcal{Z}(\theta)}\exp\left[-\ell^{2}_{\theta}(x|\bar{x})/2k\right]
\end{equation}
is simply the probability weight $\omega(x|\theta)$ of the distribution function (\ref{UGRTh}). It is presence here guarantees the exponential vanishing of the vector field $\mathbf{v}$ on the boundary $\partial\mathcal{M}$. By definition, the divergence of the vector field $\mathbf{v}$ is expressed throughout the covariant differentiation $D_{i}$ as:
\begin{equation}
\mathbf{\mathrm{div}}(\mathbf{v})\equiv D_{i}\upsilon^{i}(x|\theta).
\end{equation}
Considering definition (\ref{aux.vector}), this last expression can also be rewritten as follows:
\begin{eqnarray}\label{divvv}
\mathbf{\mathrm{div}} (\mathbf{v}) = \frac{1}{\mathcal{Z}(\theta)}\exp\left[-\eta^{2}(x|\theta)/2k\right]\left[D_{i}w^{i}(x|\theta)
+\frac{1}{k}\eta_{i}(x|\theta)w^{i}(x|\theta)\right].
\end{eqnarray}
Here, it was considered the relation:
\begin{equation}
D_{i}\eta^{2}(x|\theta)=-2\eta_{i}(x|\theta),
\end{equation}
which follows from the expression $\eta^{2}(x|\theta)=g^{ij}(x|\theta)\eta_{i}(x|\theta)\eta_{j}(x|\theta)$ and the identities:
\begin{equation}
D_{i}\eta_{j}(x|\theta)=D_{i}D_{j}\mathcal{S}(x|\theta)=-g_{ij}(x|\theta)\mbox{ and }D_{i}g^{jk}(x|\theta)=0.
\end{equation}
Result (\ref{fluct.rel}) is obtained from equation (\ref{divvv}) by performing the volume integration over the manifold $\mathcal{M}$. Considering the divergence theorem:
\begin{equation}\label{div.theorem}
\int_{\mathcal{A}} \mathbf{\mathrm{div}} (\mathbf{v})d\mu=\oint_{\partial\mathcal{A}}\mathbf{v}\cdot d\mathbf{\Sigma},
\end{equation}
one verifies the vanishing of the volume integral over the divergence $\mathbf{\mathrm{div}} (\mathbf{v})$. Precisely, the surface integral in equation (\ref{div.theorem}) vanishes when the subset $\mathcal{A}\subset\mathcal{M}$ is extended to the manifold $\mathcal{M}$. This is a consequence of the exponential vanishing of the auxiliary vector field $\mathbf{v}$ on the boundary $\partial\mathcal{M}$.

Invariant fluctuation theorem (\ref{fluct.rel}) allows us to obtain other invariant fluctuation relations. Let us consider the vector field associated with the generalized restituting forces, $w^{i}(x|\theta)=\eta^{i}(x|\theta)=g^{ij}(x|\theta)\eta_{j}(x|\theta)$. One obtains by direct differentiation the following relation:
\begin{equation}
D_{i}\eta^{i}(x|\theta)=-n,
\end{equation}
with $n$ being the dimension of the manifold $\mathcal{M}$. Combining this result with identity (\ref{fluct.rel}), one obtains the expectation value of the square of restituting generalized forces:
\begin{equation}\label{fluct.rel1}
\left\langle\eta^{2}(x|\theta)\right\rangle=nk.
\end{equation}
Considering the identities (\ref{ident.genforsep}) and (\ref{deltaS}), one obtains the expectation values:
\begin{equation}\label{fluct.rel2}
\left\langle\ell^{2}_{\theta}(x|\bar{x})\right\rangle=nk\mbox{ and }\left\langle\delta\mathcal{S}(x|\theta)\right\rangle= -nk/2.
\end{equation}
The invariant fluctuation relation (\ref{fluct.rel1}) admits the following generalization:
\begin{equation}\label{other.fluct}
\left\langle\left[\eta^{2}(x|\theta)\right]^{s}\right\rangle=(2k)^{s}\frac{\Gamma\left(s+\frac{n}{2}\right)}{\Gamma\left(\frac{n}{2}\right)},
\end{equation}
where $s$ is a positive integer and $\Gamma(x)$ is the gamma function. Considering the vector field $w^{i}(x|\theta)=\left[\eta^{2}(x|\theta)\right]^{s-1}\eta^{i}(x|\theta)$ into the identity (\ref{fluct.rel}), one obtains the following recurrence equation:
\begin{equation}\label{recurrence}
\left\langle\left[\eta^{2}(x|\theta)\right]^{s}\right\rangle=2k\left(s-1+\frac{n}{2}\right)
\left\langle\left[\eta^{2}(x|\theta)\right]^{s-1}\right\rangle.
\end{equation}
Identity (\ref{other.fluct}) is obtained as solution of equation (\ref{recurrence}) considering the particular case (\ref{fluct.rel1}) with $s=1$ and the known property of the gamma function $\Gamma(x+1)=x\Gamma(x)$.

\subsection{Asymptotic formulae}

Gaussian distribution (\ref{gaussian.approx}) constitutes a good approximation for the fluctuating behavior of thermodynamic systems with a very large number of constituents \cite{Reichl}. However, gaussian approximation fails during the occurrence of phase transitions and critical phenomena. Moreover, this distribution is unable to describe the fluctuating behavior of \emph{non-extensive systems}, such as the mesoscopic systems and the systems with long-range interactions. The macroscopic properties of these systems are highly driven by correlations that involve all system constituents. The curvature tensor of fluctuation geometry allows a better description of these situations.

Let us start this analysis considering the case of closed systems. The normalization constant $A$ that appears in Einstein's postulate (\ref{EP}) is omitted in its covariant generalization (\ref{cov.EP}). This convection guarantees the vanishing of the equilibrium value of entropy $\mathcal{S}(\bar{x}|\theta)$ if the manifold $\mathcal{M}$ is Euclidean. According to \textbf{Theorem 3} obtained in the previous section, the entropy $\mathcal{S}(\bar{x}|\theta)$ evaluated at the equilibrium state $\bar{x}$ is estimated as follows:
\begin{equation}\label{SR}
\mathcal{S}(\bar{x}|\theta)\simeq k^{2}R(\bar{x}|\theta)/6
\end{equation}
if the curvature scalar $R(\bar{x}|\theta)$ is sufficiently small. The estimation formula (\ref{SR}) considers the contribution of the second-order geometric expansion (\ref{asymptotic}) of the exact distribution function (\ref{UGRTh}). The curvature scalar $R(\bar{x}|\theta)$ allows to introduce a criterium for the applicability of the gaussian approximation (\ref{gaussian.approx}). From a geometrical viewpoint, a relevant statistical notion here is the \emph{curvature radius} $\ell_{c}=1/\sqrt{R(\bar{x}|\theta)}$. The curvature radius defines a $(n-1)$-sphere $\mathbb{S}^{(n-1)}(\bar{x}|\ell_{c})$ centered at equilibrium state $\bar{x}$ where gaussian approximation (\ref{gaussian.approx}) is applicable:
\begin{equation}
\ell^{2}_{\theta}(x,\bar{x})<\ell^{2}_{c}.
\end{equation}
Accordingly, gaussian approximation (\ref{gaussian.approx}) fully describes the system fluctuating behavior if the square of the curvature radius $\ell^{2}_{c}$ is larger than the the expectation value of the square separation distance $\left\langle\ell^{2}_{\theta}(x,\bar{x})\right\rangle$. This condition can be expressed in terms of the curvature scalar as follows:
\begin{equation}\label{crit1}
nkR(\bar{x}|\theta)< 1.
\end{equation}
Alternatively, the licitness (or failure) of gaussian approximation (\ref{gaussian.approx}) can be characterized in terms of the \emph{correlation length} $\xi$ (do not confuse this quantity with a complete set of random quantities $\xi$). For example, let us consider a system of volume $V$ near a critical point. Denoting by $d$ the spatial dimensionality of the system, gaussian approximation is applicable if the \emph{correlation volume} $v_{c}=\xi^{d}$ is smaller than the volume of the system $V$:
\begin{equation}\label{crit2}
\xi^{d}\ll V.
\end{equation}
Although the curvature scalar $R(\bar{x}|\theta)$ and the correlation length $\xi$ are different concepts, they could be associated in some way\footnote{In the framework of statistical thermodynamics, the curvature scalar $R$ of inference geometry is related to the correlation length $\xi$ by the following asymptotic expression $R\sim \xi^{d}$ \cite{Janyszek}. This type of relationship is referred to as a \emph{hyperscaling relation} in the theory of critical phenomena. It is natural to expect that an analogous result should exist for Riemannian extension of Einstein's fluctuation theory.}.

Let us now consider the case of open systems. Most of applications of statistical mechanics refer to systems that are found under the thermodynamic influence of the natural environment. Conventionally, such equilibrium situations are described within Boltzmann-Gibbs distributions (\ref{BG}). These statistical ensembles can be derived from Einstein's postulate (\ref{EP}) or its generalization (\ref{cov.EP}) as a particular asymptotic case, specifically, when the internal thermodynamic state of the environment is unaffected by the influence of the system. Although the results derived for these equilibrium situations have not a general applicability, they can be useful for practical purposes. Considering the invariant volume element (\ref{kB.rephrasing}), the statistical ensemble (\ref{BG}) can be rephrased as follows:
\begin{equation}
dp(x|\theta)=\frac{1}{Z(\theta)}\exp\left\{\left[-\theta_{i}x^{i}+s(x|\theta)\right]/k\right\}d\mu(x|\theta),
\end{equation}
where the coordinates $x=(U,O)$ are the internal energy $U$ and the generalized displacements $O=(V,\mathbf{M},\mathcal{M},\ldots)$, while the control parameters $\theta=(1/T, w/T)$ are the inverse temperature and the ratio among the generalized forces $w=(p, -\omega, -\mathbf{H},\ldots)$ and the temperature $T$. Hereinafter, the scalar function $s(x|\theta)$ is referred to as the \emph{entropy of the open system}. This function is directly associated with the density of states $\Omega(x)$ via the metric tensor $g_{ij}(x|\theta)$:
\begin{equation}
\exp\left[s(x|\theta)/k\right]\sqrt{\left|\frac{g_{ij}(x|\theta)}{2\pi k}\right|}\equiv \Omega(x).
\end{equation}
Noteworthy that the entropy $s(x|\theta)$ is not an intrinsic property of the open system. Certainly, this entropy also depends on the metric tensor $g_{ij}(x|\theta)$, which accounts for the underlying environmental influence. Formally speaking, the entropy $\mathcal{S}(x|\theta)$ of the closed system (open system + environment) can be expressed as follows:
\begin{equation}
\mathcal{S}(x|\theta)\equiv P(\theta)-\theta_{i}x^{i}+s(x|\theta),
\end{equation}
with $P(\theta)$ being the \emph{Planck thermodynamic potential}:
\begin{equation}
P(\theta)=-k\log Z(\theta).
\end{equation}
Direct application of the estimation formula (\ref{SR}) yields the following result:
\begin{equation}\label{PP.Leg}
P(\theta)\simeq\theta_{i}\bar{x}^{i}-s(\bar{x}|\theta)+k^{2}R(\bar{x}|\theta)/6.
\end{equation}
This last formula exhibits a very simple interpretation. Gaussian or zeroth-order approximation:
\begin{equation}\label{Legendre}
P(\theta)\simeq \bar{P}(\theta)=\theta_{i}\bar{x}^{i}-s(\bar{x}|\theta).
\end{equation}
is just the known \emph{Legendre transformation} that estimates the Planck thermodynamic potential $P(\theta)$ from the entropy of the open system $s(x|\theta)$. The curvature scalar $R(\bar{x}|\theta)$ introduces a correction of second-order of this transformation.

\section{Final remarks}\label{final}

Fluctuation geometry is a mathematical approach that establishes a direct correspondence among the statistical properties of a family of continuous distributions (\ref{DF}) and the notions of Riemannian geometry. In particular, the distance notion (\ref{fluct.dist}) of fluctuation geometry provides an invariant measure of the occurrence probability. Moreover, the curvature tensor of the manifold $\mathcal{M}$ accounts for the existence of irreducible statistical correlations. In accordance with asymptotic formula (\ref{asymptotic}), this geometric notion also quantifies the deviation of a given distribution function from the properties of gaussian distributions. The present geometric approach enable us to obtain information about the statistical models without special reference to any coordinate representation of the manifold $\mathcal{M}$, that is, to perform a coordinate-free treatment.

The possibility to perform a coordinate-free treatment is closely related to the requirement of general covariance of physical theories such as general relativity. Since the statistical correlations can be related to \emph{effective physical interactions}, the curvature tensor of fluctuation geometry can represent a fundamental tool in any physical theory with a statistical formulation. In particular, this notion plays a relevant role in the framework of Riemannian extension of Einstein's fluctuation theory \cite{Vel.GEO,Vel.FG}. This development leads to a geometric reinterpretation of thermo-statistical properties of a closed system. The curvature tensor of fluctuation geometry has been employed to introduce a criterion (\ref{crit1}) for the licitness of gaussian approximation. For the case of open systems, the same analysis allows to obtain the asymptotic formula (\ref{PP.Leg}), where curvature scalar introduces a second-order correction in Legendre transformation (\ref{Legendre}) between thermodynamic potentials. Some other results obtained in this work are the invariant fluctuation theorems, in particular, the general identity (\ref{fluct.rel}) and their associated fluctuation relations (\ref{fluct.rel1})-(\ref{other.fluct}).

Before to end this section, let us summarize some open problems of fluctuation geometry with a special mathematical and physical interest:
\begin{enumerate}
  \item A relevant question is to clarify how deep is the analogy with fluctuation geometry and inference geometry \cite{Vel.FG}. In particular, it is worth analyzing the possible relevance of a Riemannian gaussian representation:
  \begin{equation}
  dQ(\vartheta|\theta)=\frac{1}{z(\theta)}\exp\left[-\ell^{2}(\vartheta,\theta)/2\right]d\mu(\vartheta)
  \end{equation}
  for the distribution function $dQ(\vartheta|\theta)$ of the efficient unbiased estimators. Here, $\ell(\vartheta,\theta)$ denotes the separation distance (the arc-length of the geodesics that connects the points $\vartheta$ and $\theta\in\mathcal{P}$) associated with the distance notion of inference geometry (\ref{inf.dist}). Moreover, the quantity $d\mu(\varphi)=\sqrt{|\mathfrak{g}_{\alpha\beta}(\vartheta)/2\pi|}d\vartheta$ denotes the invariant volume element of the Riemannian manifold $\mathcal{P}$, while $z(\theta)$ is a normalization constant.

  \item Some concepts of fluctuation geometry can be useful in problems of statistical estimation, e.g.: the notion of diffeomorphic representations of a given abstract distributions family (see in \textbf{Example \ref{diffeormorphic}}). A simple argument is that some coordinate representations of a distribution function exhibit a more convenient mathematical form than other coordinate representations. This feature can be employed to build statistical estimators $\hat{\theta}$ for the control parameters $\theta$.

  \item Fluctuation geometry can be applied to continuous distribution functions of quantum mechanics, as the example of equation (\ref{QM.dist}). I think that the role of curvature as a measure of irreducible statistical correlations could be useful to characterize some quantum behaviors such as entanglement and non-locality \cite{APeres}.

  \item Some consequences of Riemannian extension of Einstein's fluctuation theory could be tested in some concrete models, e.g.: the asymptotic formula (\ref{PP.Leg}). A special interest deserves those systems that undergo the occurrence of phase transitions and critical phenomena, where gaussian approximation of Einstein fluctuation theory is expected to fail \cite{Reichl}.

\end{enumerate}

\section*{Acknowledgements}
Velazquez thanks the financial support of CONICyT/Programa Bicentenario de Ciencia y Tecnolog\'{\i}a PSD
\textbf{65} (Chilean agency).

\end{document}